\documentclass[11pt]{article}
\usepackage{csquotes}
\usepackage[english]{babel}
\usepackage{listings}
\usepackage{xcolor}
\usepackage{indentfirst}

\usepackage[letterpaper,top=2cm,bottom=2cm,left=3cm,right=3cm,marginparwidth=1.75cm]{geometry}
\usepackage[
  backend=biber,
  style=alphabetic,   
  sorting=nyt,
  maxalphanames=5,
  minalphanames=4,
  maxbibnames=10
]{biblatex}

\usepackage{mdframed}

\usepackage{amsmath}
\usepackage{amssymb}
\usepackage{amsmath,amssymb}
\usepackage[most]{tcolorbox}
\usepackage{amsthm}
\usepackage{quantikz}
\usepackage{bm}
\usepackage{bbm}
\usepackage{booktabs}   
\usepackage{graphicx}
\usepackage{braket}
\usepackage{comment}
\usepackage{xcolor}
\usepackage{tikz-cd}
\usepackage[colorlinks=true, allcolors=blue]{hyperref}
\usepackage{changepage}
\usepackage{cleveref}
\usepackage{url}

\newtheoremstyle{indented}
  {\topsep}   
  {\topsep}   
  {\itshape}  
  {0em}     
  {\bfseries} 
  {.}         
  {.5em}      
  {}          
\theoremstyle{indented}
\newtheorem{definition}{Definition}[section]
\newtheorem{theorem}{Theorem}[section] 
\newtheorem*{theorem*}{Theorem}

\newtheorem*{corollary*}{Corollary}
\newtheorem{remark}{Remark}[section]
\newtheorem{fact}[theorem]{Fact} 
\newtheorem{lemma}[theorem]{Lemma} 

\newtheorem{corollary}[theorem]{Corollary} 

\newtheorem{proposition}[theorem]{Proposition} 
\usepackage{caption}
\usepackage[shortlabels]{enumitem}
\usepackage{mathtools}

\usepackage{multirow}
\usepackage{upgreek}

\newcommand{\complex}{{\mathbb C}}

\newcommand{\finitefield}{{\mathbb F}}

\newcommand{\norm}[1]{\left\| #1 \right\|}

\newcommand{\size}[1]{\left| #1 \right|}

\newcommand{\abs}[1]{\left| #1 \right|}

\DeclareMathOperator{\Span}{span}
\DeclareMathOperator{\expct}{{\mathbb E}}

\newcommand{\transpose}{{\mathrm T}}

\newcommand{\complexi}{{\mathrm{i}}}

\newcommand{\id}{{\mathbb I}}

\newcommand{\linear}{{\mathsf L}}

\newcommand{\generallinear}{\mathsf{GL}}

\newcommand{\ketbra}[2]{| #1 \rangle\!\langle #2 |}
\newcommand{\density}[1]{\ketbra{#1}{#1}}

\newcommand{\tensor}{\otimes}

\newcommand{\adjoint}{*}

\newcommand{\directsum}{\oplus}

\newcommand{\suppress}[1]{}

\newcommand{\email}[1]{\href{mailto:#1}{#1}}

\newcommand{\hG}{\widehat{G}}
\newcommand{\qft}{{\mathrm F}}
\newcommand{\hQ}{\widehat{Q}}
\newcommand{\hT}{\widehat{T}}
\newcommand{\hP}{\widehat{P}}
\newcommand{\hK}{\widehat{K}}

\suppress{
\newcommand{\Qo}{\mathrm{Q}^{1,\pub}}
\newcommand{\Qoe}{\mathrm{C}^{1,*}}

\newcommand{\Aut}{\operatorname{Aut}}
\newcommand{\MEM}{\textsc{Memb}}
\newcommand{\ShiftEq}{\textsc{ShiftEq}}
\newcommand{\shift}{\mathrm{shift}}
\newcommand{\Corr}{\mathsf{Corr}}
\newcommand{\sfH}{\mathsf{H}}

\newcommand{\pub}{\mathrm{pub}}
\newcommand{\Dom}{\mathrm{Dom}}

\newcommand{\bbF}{\mathbb{F}}

\newcommand{\bbZ}{\mathbb{Z}}
\newcommand{\bbE}{\mathop{\mathbb{E}}}

\newcommand{\bg}{\mathbf{g}}
\newcommand{\bh}{\mathbf{h}}
\newcommand{\be}{\mathbf{e}}
\newcommand{\bi}{\mathbf{i}}
\newcommand{\bs}{\mathbf{s}}
\newcommand{\bu}{\mathbf{u}}
\newcommand{\bv}{\mathbf{v}}
\newcommand{\bG}{\mathbf{G}}
\newcommand{\bH}{\mathbf{H}}

\newcommand{\tS}{\widetilde{S}}
\newcommand{\tG}{\widetilde{G}}
\newcommand{\ts}{\widetilde{s}}
\newcommand{\tbg}{\widetilde{\bg}}
\newcommand{\tbh}{\widetilde{\bh}}
\newcommand{\tw}{\widetilde{w}}
\newcommand{\tT}{\widetilde{T}}

\newcommand{\tr}{\operatorname{tr}}

\newcommand{\muhard}{\mu_{\bG_1\bG_2,\bH S}}
}

\newcommand{\Qo}{\mathrm{Q}^{1,\pub}}
\newcommand{\Qoe}{\mathrm{C}^{1,*}}

\newcommand{\Aut}{\operatorname{Aut}}
\newcommand{\MEM}{\textsc{Memb}}
\newcommand{\ShiftEq}{\textsc{ShiftEq}}
\newcommand{\shift}{\mathrm{shift}}
\newcommand{\Corr}{\mathsf{Corr}}
\newcommand{\Herm}{\operatorname{Herm}}

\newcommand{\sfH}{\mathsf{H}}

\newcommand{\pub}{\mathrm{pub}}
\newcommand{\Dom}{\mathrm{Dom}}

\newcommand{\bbF}{\mathbb{F}}

\newcommand{\bbZ}{\mathbb{Z}}
\newcommand{\bbE}{\mathop{\mathbb{E}}}

\newcommand{\bg}{\bm{g}}
\newcommand{\bh}{\bm{h}}
\newcommand{\be}{\bm{e}}
\newcommand{\bi}{\bm{i}}
\newcommand{\bs}{\bm{s}}
\newcommand{\bu}{\bm{u}}
\newcommand{\bv}{\bm{v}}
\newcommand{\bG}{\mathbf{G}}
\newcommand{\bH}{\mathbf{H}}
\newcommand{\rH}{\mathrm{H}}
\newcommand{\hrH}{\widehat{\mathrm{H}}}

\newcommand{\tS}{\widetilde{S}}
\newcommand{\tG}{\widetilde{G}}
\newcommand{\ts}{\widetilde{s}}
\newcommand{\tbg}{\widetilde{\bg}}
\newcommand{\tbh}{\widetilde{\bh}}
\newcommand{\tw}{\widetilde{w}}
\newcommand{\tT}{\widetilde{T}}

\newcommand{\tr}{\operatorname{tr}}

\newcommand{\muhard}{\mu_{\bG_1\bG_2,\bH S}}

\title{An exponential separation between entanglement-assisted and unassisted one-way quantum communication}

\author{
Ryan Anselm~\thanks{%
Department of Computer Science, University of Texas at Austin.
Email: \email{ryan.anselm@utexas.edu}~.
}
\and
Srijita Kundu~\thanks{%
Hon Hai (Foxconn) Research Institute.
Email: \email{srijita.kundu@foxconn.com.sg}~.
}
\and
Olivier Lalonde~\thanks{%
School of Computer Science, and
Institute for Quantum Computing,
University of Waterloo, Canada.
Email: \email{olalonde@uwaterloo.ca}~.
}
\and
Ashwin Nayak~\thanks{%
Department of Combinatorics and Optimization,
and Institute for Quantum Computing, 
University of Waterloo, Canada.
Email: \email{academic@ashwinnayak.info}~.
}
}

\date{}

\begin{document}
\maketitle

\begin{abstract}
A longstanding question in quantum communication complexity is whether some task can be accomplished with a small amount of communication in the presence of entanglement, yet require much more quantum communication in the absence of entanglement. Separations of this nature were previously known for relational problems and, in the simultaneous message passing model, for partial functions. But it has remained unresolved whether any such separation exists for a total Boolean function. We resolve this question with an exponential separation in the one-way setting: we exhibit a family of total Boolean functions $f_n\colon \{0,1\}^n \times \{0,1\}^n \to \{0,1\}$ that can be computed with $O(\log n)$ bits of one-way classical communication given prior entanglement, but that require $\Omega(n^{1/3})$ qubits of one-way quantum communication without entanglement. Our function is a special case of the subgroup membership problem, first studied in the communication setting by Aaronson, Le Gall, Russell, and Tani. 
\end{abstract}

\section{Introduction}
The no-communication theorem says that entanglement cannot, by itself, be used to transmit information. No matter how much entanglement Alice and Bob share, Alice cannot send Bob a message by acting on her share alone. However, this says little about how much entanglement can help when communication is allowed but comes at a cost per bit (or qubit) transmitted. This leads to one of the central questions of quantum communication complexity, raised early on by Buhrman and de Wolf~\cite{buhrman2001communication} and by Brassard~\cite{brassard2003quantum}:
\begin{center}
    \textit{When can entanglement help save on (quantum) communication, and by how much?}
\end{center}
We show there can be \emph{exponential} savings for jointly computing a total Boolean function when communication is restricted to a single message from Alice to Bob. Bob's input specifies a yes-or-no question about Alice's input, but Alice must send her message without knowing the question. With entanglement shared before either input is known, a short classical message lets Bob answer correctly with high probability. Without prior entanglement, Alice must send an exponentially longer message to achieve the same accuracy, even if her message is quantum.

\subsection{Background}
Communication complexity studies the amount of communication required to solve a given distributed task. The canonical task takes the following form: given a function $f\colon \{0,1\}^n \times \{0,1\}^n \to \{0,1\}$, Alice receives an input $x\in \{0,1\}^n$, Bob receives an input $y \in \{0,1\}^n$, and one of them must output $f(x,y)$. The communication complexity of $f$ is the minimum number of bits that Alice and Bob need to exchange, in the worst case, to achieve this task. Trivially, the communication complexity of $f$ is at most $n$ bits by a protocol which sends one party's entire input, though this is often far from optimal. A central goal of the field is to determine when substantially less communication suffices, and to prove lower bounds showing when it does not.

We will compare two standard quantum analogs of Yao's classical communication model \cite{yao1979complexity}. In the first, which we call the \emph{unassisted quantum communication model}, Alice and Bob can communicate by sending quantum messages, but do not share any entanglement at the outset of the protocol. This model was introduced by Yao~\cite{yao1993quantum}. In the second, which we call the \emph{entanglement-assisted classical communication model}, Alice and Bob communicate by sending classical messages, but are allowed to share an entangled state of their choosing, independent of the inputs, at the outset of the protocol. This model was introduced by Cleve and Buhrman~\cite{Cleve_1997}. Note that quantum teleportation (\cite{bennett1993teleporting}) shows that any protocol in the first quantum model can be simulated in the second at the cost of doubling the amount of communication. Unless stated otherwise, all the communication models we consider tolerate error probability at most $1/3$ on every input, and allow Alice and Bob to use shared randomness. Shared entanglement subsumes shared randomness, since locally measuring EPR pairs produces shared random bits. Newman's theorem (\cite{newman91}) shows that shared randomness can be removed from communication protocols with only $O(\log n)$ additional communication and an arbitrarily small constant increase in error. However, it is known that Newman's proof technique does not apply to entanglement-assisted protocols (\cite{gavinsky2006rolesharedentanglement, jain2008optimaldirectsumprivacy}).

Although exponential quantum–classical communication separations have long been known (\cite{raz1999exponential, baryossef04, gavinsky2008exponentialseparationsonewayquantum, klartag2010quantumonewaycommunicationexponentially}), the extent to which shared entanglement can reduce quantum communication is a longstanding open problem (\cite{brassard2003quantum, buhrman2001communication}). Superdense coding, introduced by Bennett and Wiesner~\cite{bennett92}, gives a factor-of-$2$ savings for transmitting classical information using quantum messages, but no entanglement advantage stronger than this was previously known for total functions. A recurring obstacle is that many standard quantum communication lower-bound techniques apply equally well for entanglement-unassisted quantum protocols as for entanglement-assisted ones (\cite{linial07}), and therefore cannot establish such an advantage. 

Known exponential separations between the two models are generally confined to restricted settings. One such example is due to Gavinsky, Kempe, Regev, and de Wolf~\cite{gavinsky2005boundederrorquantumstateidentification}, who gave an exponential separation between the two models in the simultaneous message passing (SMP) setting for a relational problem. More recently, Arunachalam and Girish~\cite{arunachalam2025tradeoffsentanglementcommunication} obtained such a separation in the SMP setting for a partial function. However, these separations both collapse when direct (one-way or two-way) quantum communication between Alice and Bob is permitted, for the simple reason that their efficient entanglement-assisted protocols use roughly as much shared entanglement as communication. Permitting Alice and Bob to directly exchange quantum messages thus lets them establish sufficient shared entanglement themselves and then run the efficient entanglement-assisted protocol. In a different vein, Hasegawa, Le Gall, and Modanese~\cite{hasegawa2026maximumseparationquantumcommunication} recently showed a separation between the two models for a relational task with distributed outputs. Their task can be solved with zero communication given prior entanglement, but requires $\Omega(n)$ qubits of communication without prior entanglement, even when arbitrary two-way quantum communication is allowed. Despite progress along different dimensions of the problem, it remained open whether entanglement-assisted classical protocols can substantially outperform unassisted quantum protocols for functional problems when Alice and Bob communicate directly.

Another longstanding problem is the quest for strong quantum--classical communication separations for \emph{total} functions. The first such separation was shown for the disjointness function by Buhrman, Cleve, and Wigderson~\cite{buhrman1998quantumvsclassicalcommunication}, and was later sharpened to a quadratic separation by Aaronson and Ambainis~\cite{aaronson03spatialsearch}. Anshu et al.~\cite{Anshu_2016} subsequently obtained the first superquadratic separation, with a nearly $2.5$-power gap. Later work on quantum and randomized query complexity by Bansal and Sinha~\cite{bansal2020kforrelationoptimallyseparatesquantum} yielded, via query-to-communication lifting, a nearly cubic quantum--classical communication separation for a total function. However, until recently, no superpolynomial separation was known for any total function. Moreover, no polynomial separation for a total function was known in which the quantum protocol used only a constant number of rounds. Both of these gaps were recently closed by Gavinsky~\cite{gavinsky2026quantumcommunicationcomplexitytotal}, who gave a total function admitting a two-round quantum protocol with polylogarithmic communication, while every randomized classical protocol requires polynomial communication, even with arbitrarily many rounds. Subsequently, Le Gall~\cite{gall2026improvedseparationsquantumclassical} generalized and quantitatively strengthened Gavinsky's separation. Independently and concurrently with Gavinsky, Hasegawa and Le Gall~\cite{hasegawa2026constant} gave a total function with a polynomial quantum--classical separation in which the quantum protocol uses only a constant number of rounds. Whether any asymptotic quantum--classical separation for a total function is possible when the quantum protocol is one-way remains an open question.\footnote{Without shared randomness, the equality function admits a factor-of-$2$ quantum--classical separation using a protocol of Winter~\cite{winter2004quantum}.}

\subsection{Our results}
\label{sec-results}
We show that shared entanglement can give an exponential reduction in the one-way communication cost of computing a total function. Our main result is the following:
\begin{theorem}\label{thm:maintheorem}
    There exists a family of total Boolean functions
    \[
        f_n : \{0,1\}^n \times \{0,1\}^n \to \{0,1\}
    \]
    such that $f_n$ admits an entanglement-assisted one-way classical
    protocol from Alice to Bob with communication cost $O(\log n)$
    and worst-case error at most $1/3$. The protocol we construct uses
    $\Theta(n)$ shared EPR pairs. In contrast, without shared
    entanglement, the one-way quantum and randomized communication
    complexities of $f_n$, with the same error bound, are both
    $\Theta(n^{1/3})$.
\end{theorem}
Theorem~\ref{thm:maintheorem} gives the first known asymptotic separation between entanglement-assisted classical communication and unassisted quantum communication for a functional problem, thus answering the question of \cite{brassard2003quantum, buhrman2001communication} in the one-way setting. Moreover, the separation is exponential and also holds for a total function. In particular, it rules out a direct shared-entanglement analogue of Newman's theorem, even for total functions. 

Our function is based on subgroup membership, a problem introduced by Watrous~\cite{Watrous00-quantum-proofs} and ported to the one-way communication setting by Aaronson, Le Gall, Russell, and Tani \cite{Aaronson_2011}. In the one-way communication version of subgroup membership, Alice receives a subgroup $H$ of a finite group $G$, Bob receives an element $g\in G$, and they must decide whether $g\in H$ or not with a single message from Alice to Bob. In Section~\ref{sec-subgroup-membership}, we introduce a variant called \emph{bounded-order subgroup membership} ($\MEM_{G,k}$) in which Alice's subgroup has order at most $k$ for some parameter $k$. We show an entanglement-assisted classical communication upper bound of $O(\log k)$ bits for this variant, while for a suitable family of groups we show that substantially more communication is required without entanglement.

Shi and Zhu~\cite{Shi_2008} showed that any entanglement-assisted quantum protocol using $Q$ qubits or bits of communication can be simulated, to any fixed constant accuracy, by an unassisted classical SMP protocol using shared randomness and $2^{O(Q)}$ bits of communication, independent of the amount of entanglement. Since any SMP protocol can be simulated by a one-way protocol, their result implies that an entanglement-assisted classical one-way protocol using $C$ bits can be simulated by a randomized one-way protocol using $2^{O(C)}$ bits. Combined with the $\Omega(n^{1/3})$ randomized one-way lower bound of Theorem~\ref{thm:maintheorem}, this implies that $C=\Omega(\log n)$, so the $O(\log n)$-bit entanglement-assisted protocol we give is asymptotically optimal. Theorem~\ref{thm:maintheorem} also implies that an exponential blowup in the simulation cost is \emph{necessary}, even for total functions. In Section~\ref{subsec:entanglement-assisted}, we sharpen the simulation theorem of Shi and Zhu, obtaining classical SMP simulations of entanglement-assisted classical protocols using \(C\) bits of communication with cost \(O(2^C)\) in the one-way case and \(O(2^{2C})\) in the general case. By tracking constants carefully, it turns out that our results imply that the former simulation is asymptotically optimal. 

Theorem~\ref{thm:maintheorem} implies that any one-way protocol for our function using $E$ shared EPR pairs and $C$ bits or qubits of one-way communication must satisfy a trade-off: $E + C = \Omega(n^{1/3})$. To see this, suppose a one-way protocol uses $E$ shared EPR pairs and $C$ bits or qubits of communication. By first establishing the shared entanglement, such a protocol can be converted into an unassisted quantum protocol using $E+C$ qubits of one-way communication, and our $\Omega(n^{1/3})$ unassisted quantum one-way lower bound then implies the trade-off. It remains open how much shared entanglement is necessary to achieve $O(\log n)$ communication for our function: our lower bound requires $\Omega(n^{1/3})$ EPR pairs, whereas our protocol uses $\Theta(n)$.

\subsection{Technical Overview}

Our separation is proven in several stages. In Section~\ref{sec-subgroup-membership}, we prove an entanglement-assisted classical one-way upper bound of $O(\log k)$ bits of communication for the bounded-order subgroup membership function, denoted $\MEM_{\Gamma,k}$, over any finite group $\Gamma$. In Section~\ref{subsec:upperboundsreduction}, we reduce the \emph{shifted equality} function, $\ShiftEq_{G,r}$, to $\MEM_{\tG,2^{r+1}}$ where $\tG$ is a larger group constructed from the finite group $G$ and $r \geq 1$. In Section~\ref{subsec:q-lb}, we prove an unassisted quantum one-way lower bound for $\ShiftEq_{G,r}$ conditional on some representation-theoretic properties of the group $G$ it uses as a parameter. Finally, in Section~\ref{sec:heis-group}, we instantiate an explicit family of groups with the properties needed for our main theorem.

The concrete instantiation of the bounded-order subgroup membership problem we introduce is inspired by the Boolean Hidden Matching problem \cite{baryossef04, gavinsky2008exponentialseparationsonewayquantum}, as presented in \cite{shi2012limits}. In a specific case of the problem as presented therein, Alice is given a string $x\in \{0,1\}^{2n}$, Bob is given a perfect matching $M$ on $[2n]$, and a string $w\in\{0,1\}^n$. We view~$M$ as an array of length~$n$, each of whose entries is a pair of indices. There is a promise that there is a $b\in\{0,1\}$, such that for all $k\in[n]$, if $M_k=(i,j)$, then $x_i\oplus x_j=b\oplus w_k$, and the goal is to determine $b$.

We would like a generalization of Boolean Hidden Matching which can be realized as an instance of the subgroup membership problem. This would enable the use of protocols for the (bounded-order) subgroup membership problem with entanglement, and moreover, the use of the Hidden Matching-like structure to prove lower bounds. The generalization we choose is the following: instead of Alice and Bob being given strings in $\{0,1\}^n$, they are given functions from $\bbZ_2^r$ to some group $G$. Specifically, Alice is given two functions $\bg_1, \bg_2\colon \bbZ_2^r \to G$, and Bob is given a single such function $\bh\colon \bbZ_2^r \to G$, and a string $s\in \bbZ_2^r$. They are required to determine if for all $x\in\bbZ_2^r$, $\bg_1(x)\bg_2(x+s)=\bh(x)$. We call this the shifted equality problem and denote it as $\ShiftEq_{G,r}$. 

The analogy to Boolean Hidden Matching can be seen by considering two labeled copies of $\bbZ_2^r$, one indexing the coordinates of $\bg_1$ and the other indexing the coordinates of $\bg_2$. For each $x\in\bbZ_2^r$, the pair $(x,x+s)$ connects the coordinate $x$ in the first copy to the coordinate $x+s$ in the second copy. Since translation by $s$ is a bijection on $\bbZ_2^r$, these pairs form a perfect matching between the two copies. The XOR operation is then replaced by the group operation of $G$. Unlike Boolean Hidden Matching, however, we impose no promise on the structure of a no-instance. Equivalently, defining the discrepancy
\[\delta(x) := \bh(x)^{-1}\bg_1(x)\bg_2(x+s), \]
we have $\ShiftEq_{G,r}((\bg_1,\bg_2),(\bh,s))=1$ precisely when $\delta(x)$ is the identity element of $G$ for every $x$. On a no-instance of $\ShiftEq_{G,r}$, the discrepancy function $\delta$ is not identically equal to the identity element of $G$. Its value may vary arbitrarily with $x$, and it may still equal the identity at some coordinates. Thus, $\ShiftEq_{G,r}$ is a total function.

Almost all of the arguments we give apply to $\ShiftEq_{G,r}$ for a general group $G$. We pick a specific group that suits our purposes only at the very end.

\paragraph{Upper bound with entanglement.}
In Section~\ref{sec-subgroup-membership}, we give an entanglement-assisted classical one-way protocol for the bounded-order subgroup membership function, denoted $\MEM_{G,k}$, using $O(\log k)$ bits of one-way classical communication and $O(\log |G|)$ shared EPR pairs. Here $G$ is a fixed finite group, Alice is given a subgroup $H\leq G$ of order at most $k$, and Bob receives $g\in G$ and must decide whether $g\in H$. 

Our entanglement-assisted protocol is derived from a QMA verification algorithm due to Watrous~\cite{Watrous00-quantum-proofs}, which was also used in the quantum one-way protocol for subgroup membership due Aaronson, Le Gall, Russell, and Tani~\cite{Aaronson_2011}. Essentially,
Alice sends the state $\ket{H}=|H|^{-1/2}\sum_{h\in H}\ket{h}$, and Bob performs a Hadamard test with the unitary $U_g:\ket{x}\mapsto\ket{xg}$ which permutes the group elements of $G$ according to group multiplication by $g$ on the right.
Since $\ket{Hg}=\ket{H}$ when $g\in H$, and $\braket{H|Hg}=0$ otherwise, the Hadamard test returns $0$ with probability $1$ or $1/2$, respectively.
Sending a constant number of copies gives a bounded-error protocol using $O(\log |G|)$ qubits.

The main difference in our protocol from theirs is that instead of preparing and transmitting a quantum state $\ket{H}$, Alice performs \emph{remote state preparation} (\cite{Bennett_2001}) to prepare a specified state on Bob's side. Unlike quantum teleportation (\cite{bennett1993teleporting}), remote state preparation requires the sender to know a classical description of the quantum state they would like the recipient to hold. Alice remotely prepares for Bob the uniform mixture over the left coset states of $H$. Let $m = [G : H]$ and let $T_1, \dots, T_m$ be the left cosets of $H$ in $G$. Then let
\[\rho = \frac{1}{m}\sum_{i=1}^m \ket{T_i} \bra{T_i}\quad \text{where}\quad \ket{T_i} = \frac{1}{\sqrt{|H|}}\sum_{t\in T_i}\ket{t}.\] 
Since the left coset states are orthogonal, $\rho$ is a
rank-$m$ flat state of dimension $|G|$, a uniform mixture of $m$ pure $|G|$-dimensional states. Remotely preparing $\rho$ approximately can be done using $O(\log |H|)$ bits of communication and $O(\log |G|)$ shared EPR pairs using a protocol for remote preparation of flat states given by \cite{kundu2026nearoptimalentanglementcommunicationtradeoffsremote}. Bob then applies the Hadamard test as in \cite{Watrous00-quantum-proofs,Aaronson_2011} to $\rho$. Repeating this a constant number of times gives a bounded-error protocol. 

Preparing a high-rank mixed state $\rho$ rather than a pure state $\ket{H}$ on Bob's side is the key difference which enables a reduction of the entanglement-assisted classical communication down from $O(\log |G|)$ bits to $O(\log |H|)$ bits. This does not affect correctness of the protocol because any left coset state of $H$ also would have sufficed as a quantum message in the protocol of \cite{Aaronson_2011}. When $|H| \leq k = \mathrm{polylog}(|G|)$, the entanglement-assisted communication cost is only $O(\log\log |G|)$, enabling the possibility of an exponential separation from unassisted quantum one-way communication, which requires at most $O(\log |G|)$ qubits due to the quantum one-way protocol described earlier. 

\paragraph{Lower bound without entanglement.} 
We develop a technique that gives strong lower bounds on the communication in unassisted quantum one-way protocols, but not for entangle-ment-assisted classical protocols. 
As previously mentioned, this has been a major technical barrier in the study of two-party quantum communication --- most known techniques apply equally well to protocols with and without entanglement-assistance.

By the easy direction of Yao's min-max principle, it suffices to bound the success probability of an entanglement-unassisted protocol with low communication under a particular distribution on Alice and Bob's inputs. Suppose Alice and Bob's inputs are chosen from a distribution $\mu_{XY}$. Further suppose that Alice sends a $C$-qubit quantum message to Bob without entanglement assistance, and Bob performs a binary POVM on the message to obtain the output. The success probability of the protocol for computing a function $f(x,y)$ is bounded from above by
\begin{equation}\label{eq:op-bound}
\max_{\{M^y_b\}} \bbE_{x \sim \mu_X} [\|\bbE_{y \sim \mu_{Y|X=x}} M^y_{f(x,y)}\|_\infty ]
\end{equation}
where the maximum is taken over collections $\{M^y_b\}_{y \in Y, b \in \{0,1\}}$ of $2^C$-dimensional POVMs (Lemma~\ref{lem:op-norm-bound}; see also~\cite{KGDGGGHJMNHA25-quantum-classical-separation,GLYZ26-qcc-search-problems}). We get this expression by observing that for a fixed measurement strategy~$\{M^y_b\}$ used by Bob, Alice can maximize the success probability by picking the top eigenvector of the operator~$\bbE_{y \sim \mu_{Y|X=x}} M^y_{f(x,y)}$ as her $C$-qubit message. If the protocol is entanglement-assisted, Bob's POVM~$\{M^y_b\}$ acts on the message register as well as his share of the entangled state. The entangled state may have arbitrarily large dimension, and the resulting optimization problem is vastly different.

We now describe a hard distribution for $\ShiftEq_{G,r}$ that leads to a sharp upper bound on the expression in \eqref{eq:op-bound}. For this we impose additional structure on $G$: we assume $G$ has a central element~$\zeta$, i.e, an element that commutes with all other elements of $G$, which has order exactly~$3$. Under the hard distribution Alice gets uniformly random elements $\bg_1, \bg_2$ and Bob gets a uniformly random offset $s$, and $\bh$ that satisfies
\[ \bh(x) = \zeta^\phi \bg_1(x)\bg_2(x+s) \qquad \forall x\in\bbZ_2^r,\]
where $\phi=0$ with probability $\frac12$, and equal to $1$ and $2$ each with probability $\frac14$. Evidently, $\ShiftEq_{G,r}((\bg_1,\bg_2),(\bh,s))=1$ if $\phi=0$, and $0$ otherwise. Note again the analogy of this hard distribution with the Boolean Hidden Matching problem --- instead of the bit $b$ which has order $2$ under the XOR operation in Hidden Matching, we are multiplying $\bg_1(x)\bg_2(x+s)$ for all $x$ by $\zeta$ which has order $3$ under the group operation.

Letting $\bh_0$ denote Bob's input in the $\phi=0$ case for a fixed $\bg_1,\bg_2, s$, and letting $M^{\bh,s}_b$ denote Bob's measurement operators, under this hard distribution the expression~\eqref{eq:op-bound} is bounded above by
\begin{equation}\label{eq:op-O_h}
\frac12+\max_{\{M^{\bh,s}_1\}}\bbE_{\bg_1,\bg_2}\left[\left\|\bbE_s\left[\frac12M^{\bh_0,s}_1-\frac14M^{\zeta \bh_0,s}_1 - \frac14M^{\zeta^2 \bh_0,s}_1\right]\right\|_\infty\right].
\end{equation}
In the above expression, we have eliminated the outcome-$0$ POVM elements by writing $M^{\bh,s}_1=\id-M^{\bh,s}_0$. Therefore the optimization is only over $2^C$-dimensional operators $M^{\bh,s}_1$ satisfying $0\leq M^{\bh,s}_1\leq \id$. Writing $T(\bg_1,\bg_2) = \bbE_s\left[\frac12M^{\bh_0,s}_1-\frac14M^{\zeta \bh_0,s}_1 - \frac14M^{\zeta^2\bh_0,s}_1\right]$, we bound
\[ \bbE_{\bg_1,\bg_2}\left[\|T(\bg_1,\bg_2)\|_\infty\right]\]
from above.

First we note that the operator norm can be bounded above by any Schatten norm: \[ \|T(\bg_1,\bg_2)\|_\infty \leq (\tr(T(\bg_1,\bg_2)^p))^{1/p}.\]
But even this Schatten norm is hard to upper bound for the matrix $T$ directly. This is because the expectation over $s$ makes $T$ a sum over many dependent terms. (The dependence in the terms arises as they are functions of the same elements $\bg_1, \bg_2$, and the $p$-th power of this sum of dependent terms quickly becomes intractable.) 
Instead, writing $T(\bg_1,\bg_2)= \bbE_s[O_s(\bh_0)]$, we approximate it by
\[ \widetilde T(\bg_1,\bg_2,\bm{\phi})= \bbE_s[O_s(\zeta^{\phi_s}\bh_0)], \]
for a uniformly random $\bm{\phi}=(\phi_s)_{s\in S}\in\mathbb F_3^S$. Conditional on $\bg_1,\bg_2$, the summands in $\widetilde T$ are independent, since the coordinates $\phi_s$ are independent and uniform. Essentially, in $\tT$, 
we ``rotate''~$\bh_0$ by~$\zeta^{\phi_s}$ where $\phi_s$ is independent and uniform for each $s$. The terms in $\tT$ are now independent, and each $O_s$ has dimension $2^C$ and bounded operator norm (due to the measurement operators $M^{\bh,s}_1$ satisfying $M^{\bh,s}_1\leq \id$). Therefore, we are able to more easily bound
\begin{equation}\label{eq:tT-bound}
\bbE_{\bg_1,\bg_2,\bm{\phi}}\left[\tr\left(\tT(\bg_1,\bg_2, \bm{\phi})^p\right)\right] \leq O\left(2^C\left(\frac{4p}{2^r}\right)^{p/2}\right).
\end{equation}
It thus remains to show that $\bbE[\tr(T(\bg_1,\bg_2)^p)]$ can be well approximated by the above expression.\footnote{One may think we can get an expression in terms of $\tT$ directly in the success probability by taking a different $\phi_s$ for every $s$ in the hard distribution, instead of the same $\phi$ for each $s$. However, this is not the case; the distinction is that $O_s(\zeta^{\phi_s}\bh_0)$ assigns its positive coefficient to $M^{\zeta^{\phi_s}\bh_0,s}_1$, whereas the correct output in the communication problem is still $1$ exactly when Bob's array equals $\bh_0$.
}

Showing that $\bbE[\tr(T(\bg_1,\bg_2)^p)]$ is well-approximated by $\bbE[\tr(\tT(\bg_1,\bg_2, \bm{\phi})^p)]$ essentially boils down to showing a property of scalar functions $w_s: G^{\bbZ_2^r}\to \mathbb{C}$ which represent the trace of powers of~$\tT$ and have $\|w_s\|_\infty \leq 1$. 
We are then left with showing that the following approximation of products of the functions $w_s$ holds 
\begin{equation}
\label{eq-trace-power-approx}
\bbE_{\bg_1,\bg_2}\left[\prod_{s\in \bbZ_2^r}w_s(\bh_0)\right] \approx \bbE_{\bg_1,\bg_2,\bm{\phi}}\left[\prod_{s\in \bbZ_2^r}w_s(\zeta^{\phi_s}\bh_0)\right].
\end{equation}

Via a hybrid argument, the approximation in \eqref{eq-trace-power-approx} can be expressed in terms of products of suitably bounded functions on a smaller set, $Q, R : G^2 \to \mathbb{C}$. For such functions, it suffices to show
\begin{equation}
\label{eq-mixing-approximation}
\bbE_{g_0, g_1, h_0, h_1 \in G}\big[Q(g_0h_0, g_1h_1)R(g_0h_1, g_1h_0) \big] \approx  \bbE_{\substack{g_0,g_1,h_0,h_1\in G \\ \phi \in \bbF_3}}\big[Q(\zeta^\phi g_0h_0, \zeta^\phi g_1h_1)R(g_0h_1, g_1h_0) \big],
\end{equation}
where~$ g_0,g_1,h_0,h_1,\phi$ are uniformly random over the respective domains.
The difference between the above two expressions can be upper bounded by
\[ \|(\id-P)K\|_\infty\|\ket{Q}\|_2\|\ket{R}\|_2\]
where $\ket{Q}$ and $\ket{R}$ are vectors in $\complex[G^2]$ representing the functions $Q$ and $R$, and $P$ and $K$ are linear operators on this space. Specifically, $K$ maps pairs of group elements $(a_0b_0,a_1b_1)$ to a combination of pairs $(a_0b_1,a_1b_0)$, while $P$ averages over the simultaneous phase rotations $(\zeta^\phi a_0b_0,\zeta^\phi a_1b_1)$ for $\phi\in\mathbb F_3$. It turns out that the operators~$P,K$ represent the transition matrices of certain random walks on Cayley graphs over~$G^2$. The norms $\|\ket{Q}\|_2$ and $\|\ket{R}\|_2$ are upper bounded due to the hypothesis, and therefore we need to bound $\|(\id-P)K\|_\infty$ in order to prove \eqref{eq-mixing-approximation}.

A bound on the spectral norm~$\|(\id-P)K\|_\infty$ can be proved by moving to the Fourier domain. The Fourier transform over $G$ sends a vector $\ket{g}$ to a superposition over all non-isomorphic irreps of $G$, with the coefficient corresponding to an irrep $\rho$ depending on its dimension $d_\rho$. Writing $\hP$ and $\hK$ for the Fourier transforms of $P$ and $K$, it turns out the irreps where $\rho(\zeta) = \id$ cancel out in $(\id-\hP)\hK$. In fact,
\[ (\id-\hP)\hK = \bigoplus_{\rho: \rho(\zeta)\neq\id}\frac1{d_\rho}\Pi_\rho\]
where each $\Pi_\rho$ is a permutation operator. Then clearly,
\[ \|(\id-P)K\|_\infty = \|(\id-\hP)\hK\|_\infty = \frac{1}{\min_{\rho: \rho(\zeta)\neq \id}\{d_\rho\}}.\]
This gives a bound of $\frac{1}{D}$ on the approximation in \eqref{eq-mixing-approximation}, which contributes an additive $O\left(2^{Cp}\frac{2^r-1}{D}\right)$ term to the $p$-th moment bound. After taking a $p$-th root, this contributes $O\left(2^C\left(\frac{2^r-1}{D}\right)^{1/p}\right)$ to the upper bound on the success probability, where $D = D(G,\zeta) = \min_{\rho:\rho(\zeta)\neq\id} d_\rho$. For suitable $D$, the upper bound on the success probability can then be completed by combining this with \eqref{eq:tT-bound} and choosing an appropriate value of $p$.

To complete the lower bound, it suffices to exhibit an explicit group~$G$ with a central element~$\zeta$ of order~$3$ in which~$D$ is large (so that the success probability is small). The group we use is the generalized Heisenberg group $\rH_{2^{2r}}$ over $\bbF_3$, for which we can show $D(\rH_{2^{2r}},\zeta) = 3^{2^{2r}}$ for a suitable choice of $\zeta$.

\subsection{Conclusion and open problems}

We show that a large reserve of shared entanglement can be a powerful resource in communication complexity, enabling exponentially more efficient classical communication protocols than what is possible using quantum one-way communication in the absence of sufficient shared entanglement. Remarkably, this separation already arises for the fundamental task of jointly computing a total Boolean function. We suggest a few directions for future research:
\begin{enumerate}
    \item Can an exponential separation be obtained for a total function between entanglement-assisted classical SMP protocols and unassisted one-way or two-way quantum protocols? Alternatively, can such a separation be obtained when the entanglement-assisted classical protocol is one-way and the unassisted quantum protocol is two-way, or when both models allow arbitrary two-way communication? It seems unlikely that these could be obtained using our approach, as the bounded-order subgroup membership problem has an efficient two-way classical protocol and doesn't seem to have an efficient entanglement-assisted protocol in the SMP setting. For separations against two-way quantum protocols, the main difficulty would likely be to develop lower-bound techniques that apply to two-way unassisted quantum communication but not to entanglement-assisted communication.
    
    \item The entanglement-assisted classical protocol in Theorem~\ref{thm:maintheorem} uses an amount of shared entanglement polynomial in the unassisted quantum communication complexity. Can an asymptotic separation (ideally exponential) between these models require substantially more entanglement? Specifically, is there a communication problem exhibiting a large separation between the entanglement-assisted and entanglement-unassisted models for which every entanglement-assisted classical protocol achieving such a separation requires an amount of shared entanglement superpolynomial in the unassisted quantum communication complexity, or even in the input size? This seems plausible in view of known results for nonlocal games (\cite{ji2022mipre}). 

    \item In the absence of shared entanglement, does an asymptotic separation between quantum and classical one-way communication complexity exist for a total function, using the subgroup membership problem or otherwise?

    \item Can a simpler total function which admits an exponential separation between entanglement-assisted classical communication and unassisted quantum communication be identified? For the bounded-order subgroup membership function specifically, can a substantially simpler family of groups be used to reproduce our separation?


\end{enumerate}

\subsection*{Acknowledgments}

O.L.\ thanks Gilles Brassard for having introduced him to the question of quantum communication complexity with versus without entanglement many years ago. We thank Dave Touchette, Richard Cleve and Daiki Suruga for useful conversations. O.L.\ and A.N.\ acknowledge the support of the Natural Sciences and Engineering Research Council of Canada (NSERC) grants ALLRP-578455-2022 and RGPIN-2023-03731. O.L is supported by a NSERC CGRS-D scholarship. R.A. is supported by an NSF Graduate Research Fellowship. 

\subsection*{AI disclosure}

Large language models such as GPT 5.4 through GPT-6 Astra and Claude Opus 4.8 through Claude Opus 5.5 were used extensively during the course of this work. The bounded-order subgroup membership problem as well as the corresponding upper bounds were discovered by the authors. The authors invested considerable effort into proving the hardness of another instance of subgroup membership with AI assistance. After many rounds of interaction, GPT-6 Astra eventually produced a concrete instantiation of the problem in the form of the shifted equality problem and the corresponding quantum lower bound. The authors independently verified, developed, and simplified the argument presented here. The paper itself was written by the authors themselves, who assume responsibility for all content. 

\section{Preliminaries}
In this section, we introduce notation, facts, and results that will be used throughout the paper. 
Given a nonempty set $S\subseteq [n]$, we write $\ket{S} \in \mathbb{C}^n$ to mean the state:
\[\frac{1}{\sqrt{|S|}} \sum_{s \in S} \ket{s}\]

\subsection{Linear algebra and probability theory}
We use the following fact about invertible matrices:
\begin{lemma} \label{lem:conditionforinvertibility}
    Let $A$ be an $n \times n$ invertible matrix with entries in a field $\mathbb{F}$ and let $u,v \in \mathbb{F}^n$. The matrix $A + uv^T$ is invertible if and only if $v^T A^{-1} u \neq -1$.
\end{lemma}
\begin{proof}
    The matrix determinant lemma gives:
    \[\det(A + uv^T) = (1+v^T A^{-1} u)\det(A)\]
    The result follows. 
\end{proof}

We will use the following fact in order to derive a moment bound for independently distributed random Hermitian matrices.
\begin{fact}[{\cite[Eq.~(4.9) in the arXiv version]{Tro15}}]\label{fact:sign-moment}
Let $A_1,\ldots,A_n$ be fixed Hermitian $d\times d$ matrices, and let $c_1,\ldots,c_n$ be independent uniform signs in $\{-1,1\}$. For every even integer $p\geq 2$,
\[
\bbE_{c_1,\ldots,c_n} \left[\tr\left(\left(\sum_{i=1}^nc_i A_i\right)^p\right)\right]
\leq
d\left(\prod_{j=1}^{p/2}(2j-1)\right)
\left\|\sum_{i=1}^n A_i^2\right\|_\infty^{p/2}.
\]
\end{fact}

\begin{corollary}\label{cor:ind-moment}
Let $A_1,\ldots,A_n$ be independent random Hermitian $d\times d$ matrices such that $\bbE[A_i]=0$ and $\|A_i\|_\infty\leq 1$ almost surely. For every even integer $p\geq 2$,
\[
\bbE\left[\tr\left(\left(\sum_{i=1}^n A_i\right)^p\right)\right]
\leq d(4pn)^{p/2}.
\]
\end{corollary}

\begin{proof}
We introduce $A'_1,\ldots,A'_n$, where each $A'_i$ is drawn from the same distribution as $A_i$ independently. Since these matrices also have zero expectation, we have
\begin{equation}\label{eq:expec-AA'}
\bbE_{A'_1,\ldots A'_n}\left[\sum_i(A_i-A'_i)\middle|A_1,\ldots, A_n\right] = \sum_i A_i.
\end{equation}
Moreover, $(A_i-A'_i)^2\leq 4I$, which means that the operator norm of each $(A_i-A'_i)$ is at most $2$.

For even $p$, the function $X\mapsto\tr(X^p)$ is convex on Hermitian matrices. Therefore, applying Jensen's inequality to \eqref{eq:expec-AA'} gives, 
\begin{align*}
\tr\left(\left(\sum_{i=1}^n A_i\right)^p\right) & = \tr\left(\left(\bbE_{A'_1,\ldots, A'_n}\left[\sum_{i=1}^n(A_i-A'_i)\middle|A_1,\ldots, A_n\right]\right)^p\right) \\
& \leq \bbE_{A'_1,\ldots,A'_n}\left[\tr\left(\left(\sum_{i=1}^n(A_i-A'_i)\right)^p\right)\middle|A_1,\ldots, A_n\right].
\end{align*}
Averaging this over $A_i$ we get,
\begin{equation}\label{eq:expec-AA'-p}
\bbE_{A_1,\ldots, A_n}\left[\tr\left(\left(\sum_{i=1}^n A_i\right)^p\right)\right] \leq \bbE_{\substack{A_1, \ldots, A_n \\ A'_1, \ldots, A'_n}}\left[\tr\left(\left(\sum_{i=1}^n(A_i-A'_i)\right)^p\right)\right].
\end{equation}

For each $i$, exchanging $A_i$ and $A'_i$ preserves their joint distribution and negates $A_i-A'_i$. These exchanges can be made by multiplying them by independent and uniform $c_1, \ldots, c_n\in\{-1,+1\}$. Since $p$ is an even integer, we have,
\begin{equation}\label{eq:AA'-c}
\bbE_{\substack{A_1, \ldots, A_n \\ A'_1, \ldots, A'_n}}\left[\tr\left(\left(\sum_{i=1}^n(A_i-A'_i)\right)^p\right)\right]
=
\bbE_{\substack{A_1, \ldots, A_n \\ A'_1, \ldots, A'_n \\ c_1, \ldots, c_n}}\left[\tr\left(\left(\sum_{i=1}^nc_i(A_i-A'_i)\right)^p\right)\right].
\end{equation}
We now condition on all the matrices and apply Fact~\ref{fact:sign-moment} to the fixed differences $A_i-A'_i$. Since these differences are Hermitian and have operator norm at most $2$,
\[
\left\|\sum_{i=1}^n(A_i-A'_i)^2\right\|_\infty\leq 4n.
\]
Consequently,
\[
\bbE_{c_1, \ldots, c_n}\left[\tr\left(\left(\sum_{i=1}^nc_i(A_i-A'_i)\right)^p\right)\right]
\leq
d\left(\prod_{j=1}^{p/2}(2j-1)\right)(4n)^{p/2}
\leq d(4pn)^{p/2},
\]
where the last inequality uses $2j-1\leq p$ for every factor in the product. Averaging the above inequality over the matrices, and combining with \eqref{eq:expec-AA'-p} and \eqref{eq:AA'-c} proves the claim.
\end{proof}

\subsection{Group theory}
\label{sec-group-theory}

A \emph{group} $G$ is a set equipped with an associative composition map $\cdot: G \times G \to G$ such that there exists an identity element $e \in G$ and every $g\in G$ has an inverse $g^{-1} \in G$ with $g \cdot g^{-1} = e$. We will usually drop the $\cdot$ for the group operation and write $g\cdot h$ as simply $gh$. In this work, all groups are assumed to be finite.

The group is said to be abelian if the composition map is commutative, in which case it will be written with the symbol '+'.

A subset $H$ of $G$ closed under the composition map, which is also a group, is called a \emph{subgroup} of $G$, denoted by $H \leq G$. The \emph{centre} of $G$, denoted $Z(G)$, is the set of all elements which commute with all elements of $G$; an element of $Z(G)$ is called \emph{central}. It is simple to verify that $Z(G)$ is a subgroup of $G$.

The \emph{left and right cosets} of a subgroup $H$ in $G$ are the sets $gH = \{gh : h \in H\}$ and $Hg = \{hg : h \in H\}$, respectively, for $g \in G$. The number of left cosets is called the index of $H$ and is written $[G:H]$; for finite $G$, it equals $|G|/|H|$.

The left cosets form a partition of $G$, as do the right cosets. In particular, we have that $H \cap gH = H \cap Hg = \emptyset$ if $g \notin H$. The fact that left-multiplication and right-multiplication are permutations implies that these sets are all of the same cardinality, which implies that $|H|$ divides $|G|$ if $G$ is finite. This result is known as Lagrange's theorem.

Given two groups $G, H$, a map $\varphi: G \to H$ is said to be a \emph{group homomorphism} if $\varphi(g_1g_2) =\varphi(g_1)\varphi(g_2)$ for all $g_1, g_2 \in G$. A bijective homomorphism $\varphi: G \to G$ is said to be an \emph{automorphism}. We will write $\Aut(G)$ for the set of automorphisms of $G$, with multiplication given by composition and identity given by the identity map. These operations make $\Aut(G)$ a group.

For a group $G$ and a set $X$, a \emph{group action} of $G$ on $X$ is a map $\star: G \times X \to X$ satisfying $e \star x = x$ and $g \star (h \star x) = (gh) \star x$ for all $g,h \in G$ and $x \in X$. A set equipped with such an action is called a $G$-set.

\begin{definition}[Shift action]
Given a group $S$ and a set $X$, the set $X^S$ consists of all functions from $S$ to $X$. The shift action of $S$ on $X^S$ is defined, for $s,t \in S$ and $f \in X^S$, by
\[\shift_s(f)(t) = f(ts).\]
When $X$ is a group $G$, this is an action of $S$ on the arrays $G^S$.
\end{definition}

We have the following lemma that identifies the shift action between two groups as an automorphism. 
\begin{lemma}\label{lem:ishomomorphism}
    Let $G$ be a group and let $S$ be an abelian group, written
    additively. For each $s \in S$, define
    $\shift_s: G^S \to G^S$ by
    \[
        \shift_s(\bg)(t) = \bg(s+t).
    \]
    Then $\shift_s \in \Aut(G^S)$ for every $s \in S$, and the
    map $s \mapsto \shift_s$ is a homomorphism from $S$ to
    $\Aut(G^S)$.
\end{lemma}

\begin{proof}
    We first check that $\shift_s$ preserves multiplication.
    For $\bg,\bh \in G^S$ and $t \in S$,
    \begin{align*}
        \shift_s(\bg\bh)(t)
        &= (\bg\bh)(s+t)\\
        &= \bg(s+t)\bh(s+t)\\
        &= \shift_s(\bg)(t)\shift_s(\bh)(t).
    \end{align*}
    Moreover, shifting by $-s$ undoes shifting by $s$ in either
    order. Thus $\shift_s$ is bijective, with inverse
    $\shift_{-s}$, and hence belongs to $\Aut(G^S)$.

    To prove the homomorphism property, we take $s,s',t \in S$
    and compute
    \begin{align*}
        \shift_s(\shift_{s'}(\bg))(t)
        &= \shift_{s'}(\bg)(s+t)\\
        &= \bg(s'+(s+t))\\
        &= \bg((s+s')+t)\\
        &= \shift_{s+s'}(\bg)(t),
    \end{align*}
    where the third equality uses that $S$ is abelian.
    Therefore $\shift_{s+s'}=\shift_s\circ\shift_{s'}$, as required.
\end{proof}

\begin{definition}[Direct products]
Given two groups $G$ and $H$, their direct product is the group $G \times H$ with multiplication defined componentwise. More generally, $G^S$ is a group under coordinatewise multiplication: for $\bg,\bh \in G^S$, we have $(\bg\bh)(s) = \bg(s)\bh(s)$ for every $s \in S$. $G^S$ can also be viewed as the set of functions from $S$ to $G$.
\end{definition}

\begin{definition}[Semidirect products]
Given two groups $G$ and $H$ along with a homomorphism $\varphi: H \to \Aut(G)$, the semidirect product of $G$ and $H$ with respect to this homomorphism, written $G \rtimes_\varphi H$, is the set $G \times H$ with multiplication given by:
\[(g,h) \cdot (g', h') = (g \cdot (\varphi(h)(g')), h \cdot h')\]
\end{definition}
It may be checked that this multiplication is associative and that inverses exist, given by the formula $(g,h)^{-1} = (\varphi(h^{-1})(g^{-1}), h^{-1})$.

\subsection{Representation theory}
\label{sec-rep-theory}
We briefly state the concepts and results from group representation theory that we use in this work. We refer the reader to a text such as~\cite{Serre77-representation-theory} for a deeper treatment of the subject.

For a complex vector space~$V$, let~$\linear(V)$ denote the set of linear operators on~$V$, and~$\generallinear(V)$ denote the set of invertible linear operators on~$V$. Recall that the set~$\generallinear(V)$ forms a group under composition of operators. 

A \emph{(linear) representation}~$\rho$ of a group~$G$ over a complex vector space~$V$ is a group homomorphism~$\rho : G \rightarrow \generallinear(V)$. We say that the \emph{dimension} or \emph{degree} of~$\rho$ is~$\dim(V)$, and denote it by~$d_\rho$.

Let~$\rho, \sigma$ be two representations of~$G$ over~$V,W$, respectively. We say that~$\rho, \sigma$ are \emph{isomorphic} and write~$\rho \sim \sigma$ if there is an invertible linear transformation~$M : V \rightarrow W$ such that~$\sigma(g) = M \rho(g) M^{-1}$ for all~$g \in G$. Every representation~$\rho$ of a finite group is isomorphic to a \emph{unitary representation}, i.e., one which satisfies~$\rho(g)^{-1} = \rho(g)^\adjoint$ for all~$g \in G$.

The \emph{direct sum} of~$\rho$ and~$\sigma$ is defined as the representation~$\tau : G \rightarrow \mathrm{GL}(V \directsum W)$ given by~$\tau(g) = \rho(g) \directsum \sigma(g)$ for all~$g$. In matrix form, $\tau(g)$ is a block-diagonal matrix with the blocks being~$\rho(g)$ and~$\sigma(g)$. A representation is called \emph{irreducible} (an \emph{irrep} in short) if it cannot be decomposed as a direct sum of two nonzero representations. Every representation of a finite group~$G$ can be decomposed into a direct sum of irreps of~$G$.

The Schur Lemma is a fundamental result in representation theory.
\begin{proposition}[Schur Lemma]
\label{prop-schur}
Let~$\rho, \sigma$ be two irreducible representations of~$G$ over~$V,W$, respectively. Let~$M : V \rightarrow W$ be a linear transformation such that~$M \circ \rho(g) = \sigma(g) \circ M$ for all~$g \in G$. Then
\begin{enumerate}[(a)]
\item If~$\rho$ and~$\sigma$ are not isomorphic, then~$M = 0$, and
\item If~$\rho = \sigma$, then~$M = \lambda \id$ for some~$\lambda \in \complex$.
\end{enumerate}
\end{proposition}

Let~$\updelta$ denote the Kronecker delta function. An important consequence of the Schur Lemma is the following orthogonality relation for irreps. Let \(\widehat G\) denote a complete set of nonisomorphic unitary irreps of \(G\), and fix an orthonormal basis for the representation space of each \(\rho\in\widehat G\).
\begin{proposition}
\label{prop-irrep-orthogonality}
Let~$\rho, \sigma \in \hG$ be two irreps of~$G$. Then for all~$i,j \in [d_\rho]$ and~$k,l \in [d_\sigma]$, we have
\[
\frac{1}{\size{G}} \sum_{g \in G} \sqrt{d_\rho d_\sigma} \, (\rho(g)(i,j))^* \, \sigma(g)(k,l) = \updelta_{\rho, \sigma} \updelta_{i,k} \updelta_{j,l}.
\]
\end{proposition}

Every finite group~$G$ has a finite number of nonisomorphic irreps.  Then
\[
\sum_{\rho \in \hG} d_\rho^2 = \size{G}.
\]
Thus, the vector space~$\complex^G$ is isomorphic to the complex vector space
\[
V(\hG) = \Span \left\{ \ket{\rho, i, j} : \rho \in \hG, ~ i,j \in [d_\rho]  \right\} \cong \bigoplus_{\rho \in \hG} \complex^{d_\rho} \tensor \complex^{d_\rho}.
\]
The Fourier Transform~$\qft_G$ over~$G$ is a linear transformation from~$\complex^G$ to~$V(\hG)$ defined by
\[
\qft_G \ket{g} = \frac{1}{\sqrt{\size{G}}} \sum_{\rho \in \hG; \, i,j \in [d_\rho]} \sqrt{d_\rho} \, \rho(g)(i,j) \ket{\rho, i, j}
\]
for all~$g \in G$. 
By \Cref{prop-irrep-orthogonality}, the operator~$\qft_G$ is unitary and gives us a canonical isomorphism between the vector spaces~$\complex^G$ and~$V(\hG)$. We omit the subscript~$G$ from~$\qft_G$ when it is clear from the context.

\section{Communication complexity}
\subsection{Basic definitions}
Given finite input sets $X, Y$ and $\Dom \subseteq X \times Y$ and a function $f: \Dom \to \{0,1\}$, communication complexity studies the amount of communication between Alice and Bob that is needed in a given model to compute $f(x,y)$ given that Alice is given $x$ and Bob is given $y$. We call $f$ a total function if $\Dom = X \times Y$ and a partial function otherwise. 

With $f$ as before, for $0 < \varepsilon < 1/2$, we write $\mathrm{C}_{\varepsilon}(f)$ to denote the minimum communication cost of a two-way protocol for $f$ using private randomness which computes $f$ correctly with success probability at least $1-\varepsilon$ for all inputs. It is customary to take $\varepsilon = 1/3$. We will add superscripts to further specify the communication model. The superscript `1' means that we restrict to one-way communication, that is to say, a single message from Alice to Bob. The superscript `pub' means that Alice and Bob are allowed to share public coins.

We will also consider another model of communication more restricted than one-way, which is the simultaneous message passing (SMP) model. In this model, Alice and Bob do not communicate with each other directly, but instead each sends a single message to a third party referee, who must compute the function value only from these messages. We will use the superscript `$\parallel$' to denote the SMP model. Shared randomness in the SMP model may be of three types: between the players only, between the players and the referee, and between all three parties. In this work, we exclusively consider the model where only Alice and Bob share randomness, and when we write $\mathrm{C^{\parallel, \mathrm{pub}}}$, we mean classical communication complexity in the SMP model with this type of shared randomness. 

A distribution $\mu_{XY}$ on $\Dom$ in the argument of $\mathrm{C}$ indicates that we minimize the communication cost over  protocols achieving success probability at least $1-\varepsilon$ on average when Alice and Bob's inputs are sampled from the distribution $\mu_{XY}$. It is simple to see that we may restrict to deterministic protocols in the distributional setting.

We have Yao's principle, which relates communication complexity in the worst-case correctness model to the distributional correctness model.
\begin{fact} \label{fact:yao}
    For any function $f$, 
    \[\mathrm{C}^{\mathrm{pub}}_{\varepsilon}(f) = \max_{\mu_{XY}} \mathrm{C}_{\varepsilon}(f,\mu_{XY})\]
    The same holds in the one-way setting.
\end{fact}

\subsection{Quantum communication complexity without entanglement}
In the first quantum communication model, Alice and Bob are allowed to exchange quantum messages, but may not share entanglement at the outset of the protocol. The communication complexity of $f$ in this model will be denoted in the same way as in the classical model with the letter `Q' being used instead of `C'. For example, $\Qo(f)$ refers to the one-way quantum communication complexity of $f$ in the presence of a public coin. Yao's principle (Fact \ref{fact:yao}) holds for any communication model where public randomness is allowed, so it holds in this model as well. That is,
\[ \Qo_\varepsilon(f) = \max_{\mu_{XY}}\mathrm{Q}^1_\varepsilon(f, \mu_{XY}).\]
Here again it is sufficient to restrict to protocols without shared randomness in the distributional setting.

We have the following useful characterization of the best success probability achievable for a distributional task by a one-way protocol in this model which uses a given amount of quantum communication. Independently of our work, Gao, Liu, Yang, and Zhang~\cite[Sec.~1.2, and Proof of Theorem~1.1, Eq.(3)]{GLYZ26-qcc-search-problems} also used this as the starting point of their lower bounds on one-way quantum communication. A similar characterization was used for \emph{classical} one-way protocols in~\cite[Eq.~(S2)]{KGDGGGHJMNHA25-quantum-classical-separation}. 

\begin{lemma}\label{lem:op-norm-bound}
    Let $X,Y$ be input sets. Take $\Dom \subseteq X \times Y$, a function $f: \Dom \to \{0,1\}$, and a distribution $\mu$ on $\Dom$. The best success probability achievable under the distribution $\mu_{XY}$ by a one-way quantum protocol without entanglement, communicating a $C$-qubit quantum system, is:
    \[\max_{\{M^y_b\}} \bbE_{x \sim \mu_X} [\|\bbE_{y \sim \mu_{Y|X=x}} M^y_{f(x,y)}\|_\infty ]\]
    where the maximum is taken over collections $\{M^y_b\}_{y \in Y, b \in \{0,1\}}$ of $2^C$-dimensional POVMs, i.e. $2^C\times 2^C$ positive semidefinite matrices with $M^y_0 + M^y_1 = \id$ for all $y$.
\end{lemma}

\begin{proof}
On input $x$, Alice sends a density matrix $\rho_x$ on a $2^C$-dimensional register. Since there is no shared entanglement, Bob's output on input $y$ is described by a binary POVM $\{M^y_0,M^y_1\}$ on this register. Note that these POVMs also model attaching an ancilla to the message sent by Alice and doing a joint measurement on the message and ancilla.

The success probability of the protocol is
\begin{align*}
p_{\mathrm{succ}}
&=\bbE_{(x,y)\sim\mu_{XY}}\left[\tr\left(\rho_x M^y_{f(x,y)}\right)\right]\\
&=\bbE_{x\sim\mu_X}\left[
\tr\left(\rho_x\,\bbE_{y\sim\mu_{Y|X=x}}\left[M^y_{f(x,y)}\right]\right)
\right] \\
& \leq
\bbE_{x\sim\mu_X}\left[
\left\|\bbE_{y\sim\mu_{Y|X=x}}\left[M^y_{f(x,y)}\right]\right\|_\infty
\right]
\end{align*}
where the inequality is from the definition of operator norm. Maximizing over Bob's POVMs gives the upper bound direction.

To get that this quantity is also achievable as the best success probability, we fix any collection of binary POVMs $\{M^y_b\}$. For each $x$ with $\mu_X(x)>0$, Alice can choose a unit eigenvector $\ket{\psi_x}$ corresponding to the largest eigenvalue of
\[
\bbE_{y\sim\mu_{Y|X=x}}\left[M^y_{f(x,y)}\right].
\]
She sends $\ket{\psi_x}$, and Bob measures using $\{M^y_0,M^y_1\}$ on input $y$. The success probability conditioned on Alice's input being $x$ is then
\[
\bbE_{y\sim\mu_{Y|X=x}}\left[
\bra{\psi_x}M^y_{f(x,y)}\ket{\psi_x}
\right]
=
\left\|\bbE_{y\sim\mu_{Y|X=x}}\left[M^y_{f(x,y)}\right]\right\|_\infty.
\]
For inputs with $\mu_X(x)=0$, Alice's message can be chosen arbitrarily. Averaging over $x$ shows that the upper bound for this collection of POVMs is achievable. Maximizing over the collections proves the equality.
\end{proof}

\subsection{Entanglement-assisted classical communication complexity} \label{subsec:entanglement-assisted}

In the second quantum communication model, Alice and Bob are allowed to share an entangled state at the start of the protocol, and do quantum operations locally, but they communicate by exchanging classical messages. The entanglement is uncorrelated with the inputs, and Alice and Bob's local operations (depending on their inputs and classical transcript) during the protocol act on their parts of the entangled state. We denote the one-way version of quantum communication complexity in this model by $\Qoe_\varepsilon(f)$. We do not additionally indicate Alice and Bob sharing public randomness via the superscript `pub' here, since shared randomness can be simulated by shared entanglement. Since Alice and Bob implicitly have shared randomness, we have once again by Yao's principle,
\[ \Qoe_\varepsilon(f) = \max_{\mu_{XY}}\Qoe_\varepsilon(f, \mu_{XY}).\]

There is a third model of communication complexity one could consider, where in addition to sharing entanglement, Alice and Bob are additionally allowed to exchange quantum messages. This model turns out to be equivalent to the second model up to a factor of $2$ in the communication complexity. This is because one qubit of communication can be simulated by $2$ bits of classical communication and one shared EPR pair.

We will now prove a theorem that shows that entanglement-assisted protocols can be simulated by classical protocols with shared randomness with an exponential blowup in communication. The fact that there exists such a simulation with exponential blowup in communication was previously known: Shi and Zhu \cite{Shi_2008} showed that $C$ bits of entanglement-assisted communication can be simulated by $2^{O(C)}$ bits of classical communication in the SMP model. The result we show here gives the precise constant in big $O$ in the exponent.

In fact we will use the same technique as \cite{Shi_2008}: reducing the simulation of an entanglement-assisted protocol to the computation of the real inner product between two vectors. We use the following result for computing inner products by them.

\begin{theorem}[{\cite[Section~2]{Shi_2008}}]\label{fact:ip-est}
Let $u,v$ be real vectors in the same finite-dimensional space with $\|u\|_2,\|v\|_2$ $\leq 1$, held by Alice and Bob respectively. For any $0<\delta,\eta<1/2$, there is a classical SMP protocol, with shared randomness between Alice and Bob only, using
\[
O\left(\frac{\log(1/\eta)}{\delta^2}\right)
\]
bits of communication in which the referee obtains an estimate of $\langle u|v\rangle$ with additive error at most $\delta$, except with probability at most $\eta$.
\end{theorem}

\begin{theorem}\label{thm:c-sim}
For any fixed constant $\eta>0$, any entanglement-assisted protocol using $C$ bits of communication can be simulated by a classical SMP protocol with shared randomness using $O(2^{2C})$ bits of communication, while changing the acceptance probability on every input by at most $\eta$. If the original protocol is one-way, the simulation uses $O(2^C)$ bits.
\end{theorem}

\begin{proof}
We will show that there is a classical SMP protocol that approximates the original protocol's success probability up to any fixed $\varepsilon>0$. Our proof will consist of showing that the success probability of the original protocol can be approximated by vectors of the appropriate size, and then applying Theorem~\ref{fact:ip-est}.

For the general interactive case, \cite[Lemma~72]{KL26b}, applied with approximation error $\varepsilon/3$, gives real vectors $v_x,w_y$ satisfying $\|v_x\|_2,\|w_y\|_2\leq 1$ and
\[
\left|\Pr[b=1\mid x,y]-2^C\langle v_x|w_y\rangle\right|
\leq \frac{\varepsilon}{3},
\]
where $b$ is the original protocol's output. Alice can compute $v_x$ from $x$, and Bob can compute $w_y$ from $y$. By Theorem~\ref{fact:ip-est}, the referee can estimate $\langle v_x|w_y\rangle$ to additive error $\varepsilon/(3\cdot 2^C)$, with failure probability at most $\varepsilon/3$, using
\[
O\left(\frac{2^{2C}}{\varepsilon^2}\log\frac{1}{\varepsilon}\right)
\]
bits of communication. Multiplying this estimate by $2^C$ gives an estimate of $\Pr[b=1\mid x,y]$ with additive error at most $2\varepsilon/3$ whenever the estimation succeeds. The referee truncates the scaled inner product estimate to $[0,1]$ if it happens to lie outside $[0,1]$, then outputs a bit $b'$ whose probability of being $1$ is equal to the resulting value. Overall we get,
\[
\left|\Pr[b'=1\mid x,y]-\Pr[b=1\mid x,y]\right|
\leq \frac{2\varepsilon}{3}+\frac{\varepsilon}{3}
=\varepsilon.
\]
Since $\varepsilon$ is fixed, this proves the $O(2^{2C})$ communication bound.

We now consider a one-way protocol from Alice to Bob. We may assume that the shared state is pure, denoted by $\ket{\Psi}$, by including a purifying register on Alice's side. Alice's procedure for generating her message is described by a POVM
\[
\{M^x_z\}_{z\in\{0,1\}^C},
\qquad
\sum_{z\in\{0,1\}^C}M^x_z=I.
\]
After receiving the message $z$, Bob uses a binary POVM $\{N^{y,z}_0,N^{y,z}_1\}$. Thus
\[
\Pr[b=1\mid x,y]
=
\sum_{z\in\{0,1\}^C}
\bra{\Psi}M^x_z\otimes N^{y,z}_1\ket{\Psi}.
\]

We construct the vectors for this case directly:
\begin{align*}
\ket{v_x}
&=\sum_{z\in\{0,1\}^C}\ket{z}\otimes
\left((M^x_z\otimes I)\ket{\Psi}\right),\\
\ket{w_y}
&=\frac1{2^{C/2}}\sum_{z\in\{0,1\}^C}\ket{z}\otimes
\left((I\otimes N^{y,z}_1)\ket{\Psi}\right).
\end{align*}
Taking their inner product gives
\[
\Pr[b=1\mid x,y]=2^{C/2}\braket{v_x|w_y}.
\]

The improvement over the interactive case comes from the normalization of Alice's POVM: in the interactive case, her vectors also need to be normalized by the $2^{-C/2}$ factor. Here since $0\preceq M^x_z\preceq I$, we have $(M^x_z)^2\preceq M^x_z$, and consequently
\[
\|v_x\|_2^2 =\sum_z\bra{\Psi}(M^x_z)^2\otimes I\ket{\Psi} \leq \bra{\Psi}\left(\sum_z M^x_z\right)\otimes I\ket{\Psi}
=1.
\]
For Bob's vector, each $N^{y,z}_1$ satisfies $0\preceq N^{y,z}_1\preceq I$, so
\[
\|w_y\|_2^2
=
2^{-C}\sum_z\bra{\Psi}I\otimes(N^{y,z}_1)^2\ket{\Psi}
\leq 2^{-C}\sum_z 1
=1.
\]
If these vectors have complex entries, replacing each vector by the concatenation of its real and imaginary parts preserves its norm and the real part of its inner product with another vector. The inner product above is real, so this gives real vectors with the same norms and the same acceptance-probability identity.

It now suffices to estimate this inner product to additive error $\varepsilon/(3\cdot 2^{C/2})$, with failure probability at most $\varepsilon/3$. Theorem~\ref{fact:ip-est} gives communication
\[
O\left(\frac{2^C}{\varepsilon^2}\log\frac{1}{\varepsilon}\right).
\]
The referee multiplies the estimate by $2^{C/2}$ and uses the same truncation and randomized output procedure as above. The resulting acceptance probability differs from the original one by at most $\varepsilon$ on every input. For fixed $\varepsilon$, the communication is $O(2^C)$.
\end{proof}

\subsubsection{Remote state preparation}
We will be interested in the following task.
\begin{definition}[Remote preparation of flat states] \label{def:rsp}
Let $d,r \in \mathbb{N}$ with $d \geq r$ and $\varepsilon > 0$. Alice is given the classical description of a rank-$r$ projector $P$ on $\mathbb{C}^d$. Alice and Bob are to exchange classical communication so that, at the end of the protocol, there is a dimension-$d$ register $B$ in Bob's lab such that, except with probability $\varepsilon$, 
\[\|\rho_B - P/r\|_{\tr} < \varepsilon.\]
\end{definition}
This task was previously studied in various works, beginning with \cite{jain2003directsumtheoremcommunication}. The following result, due to \cite{kundu2026nearoptimalentanglementcommunicationtradeoffsremote}, describes the state-of-the-art protocol for achieving it given access to prior entanglement. 
\begin{theorem}[Theorem 1.2 of \cite{kundu2026nearoptimalentanglementcommunicationtradeoffsremote}] \label{thm:rspprotocol}
    It is possible to achieve the task in Definition \ref{def:rsp} given that Alice and Bob share $\lceil \log_2 d \rceil$ EPR pairs, with
    \[\log({d/r}) + O(\log(1/\varepsilon))\]
    bits of one-way communication from Alice to Bob. 
\end{theorem}

\section{Bounded-order subgroup membership}
\label{sec-subgroup-membership}

The function we consider is based on the subgroup membership problem, which was first studied by Watrous~\cite{Watrous00-quantum-proofs} in the context of quantum complexity theory. A distributed version of the problem was studied as a candidate total function for a quadratic separation between one-way quantum and classical communication complexity by Aaronson, Le Gall, Russell, and Tani~\cite{Aaronson_2011}. Our variant adds the restriction that the subgroups are guaranteed to have order at most $k$.
\begin{definition}[Bounded-order subgroup membership]
\label{def:bounded-subgroup}
Let $G$ be a group. For $1\leq k \leq |G|$, let $\mathcal{H}_{G,k}$ be the set of subgroups of $G$ of order at most $k$. Define the bounded-order subgroup membership function $\MEM_{G,k}: \mathcal{H}_{G,k} \times G \to \{0,1\}$ by
\[\MEM_{G,k}(H, g) = \begin{cases}
    1 & \textrm{if } g \in H,\\
    0 & \textrm{if } g \notin H.
\end{cases}\]
\end{definition}
In the two-party communication setting, the subgroup~$H$ is given to Alice, and the group element~$g$ is given to Bob. Subgroup membership has a simple classical one-way protocol where Alice sends the label of her subgroup, which uses $O(\log^2{|G|})$ bits of communication because any group $G$ has at most $2^{(\log_2 |G|)^2}$ subgroups. This follows from the fact that any subgroup of $G$ is generated by at most $\log_2 |G|$ elements. It also admits a quantum one-way protocol using $O(\log |G|)$ qubits where Alice sends $\ket{H}$, a uniform superposition over elements of her subgroup, and Bob performs a Hadamard test on $\ket{H}$ with multiplication by the element $g$~\cite{Watrous00-quantum-proofs}. (See \Cref{sec-entanglement-assisted-protocol} for the details.)

Aaronson, Le Gall, Russell, and Tani~\cite{Aaronson_2011} showed that the quantum one-way communication complexity of subgroup membership (without a restriction on the subgroup order) is $\Omega(\log |G|)$ for a variety of groups, such as $G = \bbZ_2^r$ for $r\geq 1$. Their lower bound also holds in the presence of shared entanglement, so without a bound on the subgroup order, entanglement gives no asymptotic advantage for these groups. With an upper bound on the subgroup order however, we show in \Cref{sec-entanglement-assisted-protocol} that subgroup membership has an efficient entanglement-assisted protocol for all finite groups.

\begin{theorem} \label{thm:upperboundsubgroupmembership}
    For all finite groups $G$ and $1 \leq k \leq |G|$, there exists an entanglement-assisted one-way communication protocol that uses $2\lceil\log_2 |G|\rceil$ shared EPR pairs and $O(\log k)$ bits of communication and computes $\MEM_{G,k}$ with error at most $1/3$. In particular,
    \[C^{1, *}_{1/3}(\MEM_{G,k}) = O(\log k).\]
\end{theorem}

In Section~\ref{sec:shifteq}, we present a problem that reduces to bounded-order subgroup membership on a particular group $\tG$, and whose one-way quantum and classical communication complexities without entanglement are exponentially larger than its entanglement-assisted one-way communication complexity.
We also show in \Cref{sec-unassisted-protocols} that for abelian groups, the upper bound of Theorem~\ref{thm:upperboundsubgroupmembership} can be matched without shared entanglement.
\begin{theorem}
    \label{thm:abelian_subgroup}
    Let $G$ be a finite abelian group and $1\leq k\leq |G|$. Then
    \[C^{1,\mathrm{pub}}_{1/3}(\MEM_{G,k}) = O(\log k).\]
\end{theorem}
Thus, $\MEM_{G,k}$ cannot yield an asymptotic quantum-classical separation if $G$ is abelian.

For completeness, we give upper bounds on the classical communication complexity of the bounded-order subgroup membership function without shared entanglement in the two-way and simultaneous message passing communication models in \Cref{sec-unassisted-protocols}.
\begin{theorem}
    \label{thm:subgroup_classical_other}
    For all finite groups $G$ and $1\leq k\leq |G|$,
    \[C^{\mathrm{pub}}_{1/3}(\MEM_{G,k}) = O(\log k), \qquad C^{\parallel, \mathrm{pub}}_{1/3}(\MEM_{G,k}) = O(k).\]
\end{theorem}

\subsection{The entanglement-assisted protocol}
\label{sec-entanglement-assisted-protocol}

We now describe the entanglement-assisted protocol which shows Theorem \ref{thm:upperboundsubgroupmembership}. This protocol is similar in spirit to the quantum one-way protocol of \cite{Aaronson_2011} but uses remote state preparation~\cite{Bennett_2001,kundu2026nearoptimalentanglementcommunicationtradeoffsremote} instead of sending a quantum message directly. Kundu and Lalonde~\cite{kundu2026nearoptimalentanglementcommunicationtradeoffsremote} show that the amount of shared EPR pairs this protocol uses for remote state preparation is essentially optimal. Alice remotely prepares two copies of a mixed state $\rho$ on Bob's side, where $\rho$ is the uniform mixture of all of the left coset states of $H$. Bob then performs a Hadamard test (see Figure~\ref{fig:hadamardtest}) on each copy of $\rho$ with a unitary $U_g$ which applies right multiplication by $g$, given by $\ket{h} \mapsto \ket{hg}$ for all $h\in G$. Bob accepts if and only if both Hadamard tests measure $0$.

\begin{figure}[h]
    \centering
    \begin{quantikz}
    \lstick{$\rho$} & \qw      & \gate{U_g}    & \qw      & \qw \\
    \lstick{$\ket{0}$}    & \gate{\sfH} & \ctrl{-1}   & \gate{\sfH} & \meter{}
    \end{quantikz}
    \caption{The Hadamard test circuit}
    \label{fig:hadamardtest}
\end{figure}

\begin{figure}[htbp]
\begin{tcolorbox}[colback=white,colframe=black,sharp corners,boxrule=0.5pt] 
\textbf{Given.}

\begin{itemize}
    \item A finite group $G$ and a parameter $1 \leq k\leq |G|$.
\end{itemize}

Alice receives $H \in \mathcal{H}_{G,k}$ and Bob receives $g \in G$. 
\medskip

Let $m = [G : H]$ and let $T_1, \dots, T_m$ be the left cosets of $H$ in $G$. Define the state $\rho$ by
\[\rho = \frac{1}{m}\sum_{i=1}^m \ket{T_i} \bra{T_i}\quad \text{where}\quad \ket{T_i} = \frac{1}{\sqrt{|H|}}\sum_{t\in T_i}\ket{t}.\] 
\medskip

\noindent\textbf{Protocol.}
\begin{enumerate}
\item Alice sends Bob the value of $|H|$. Alice and Bob then run two independent instances of the protocol from \Cref{thm:rspprotocol} with $\varepsilon=0.01$ to remotely prepare mixed states $\hat\rho_j$ for $j\in\{1,2\}$ which are approximately $\rho$ on Bob's entangled registers. The two instances use disjoint EPR registers and independent randomness.
\item Bob performs a Hadamard test (shown in Figure~\ref{fig:hadamardtest}) on each $\hat{\rho}_j$ with respect to the unitary $U_g$, getting measurement bits $b_1, b_2$. He outputs $1$ if $b_1 = b_2 = 0$, and $0$ otherwise.
\end{enumerate}

\end{tcolorbox}
\caption{An efficient entanglement-assisted protocol for $\MEM_{G,k}$ using $2\lceil \log_2 |G|\rceil$ EPR pairs and $O(\log k)$ bits of communication.}
\label{fig:theprotocol}
\end{figure}

\begin{proof}[Proof of Theorem~\ref{thm:upperboundsubgroupmembership}]
The case $k=1$ is trivial, since then $H=\{e\}$ and Bob can determine membership without communication. We henceforth assume $k\ge 2$. We refer to the protocol in Figure~\ref{fig:theprotocol}. $U_g$ applies the map $\ket{h} \mapsto \ket{hg}$ for all $h\in G$. If $g \in H$, then $Hg = H$ and so $T_ig = T_i$, but if $g \notin H$, then $Hg$ will be a distinct right coset of $H$ such that $Hg \cap H = \emptyset$. Left-multiplying both sides by any $t\in T_i$ gives that $tHg \cap t H = \emptyset$. This implies that for all left cosets $T_i$,
\[ \braket{T_i | U_g|T_i} = \begin{cases}
    1 & \text{if } g \in H,\\
    0 & \textrm{if } g \notin H
\end{cases}\quad \textrm{and so}\quad \tr(\rho U_g) =
  \frac{1}{m}\sum_{i=1}^{m} \braket{T_i|U_g|T_i} =
  \begin{cases}
    1 & \text{if } g \in H,\\
    0 & \text{if } g \notin H.
  \end{cases}\]
Condition on the event $E$ that $\lVert \hat{\rho}_j - \rho\rVert_{\tr} < 0.01$ for each $j\in \{1,2\}$. Then 
\[\lvert \tr(\hat{\rho}_j U_g) - \tr(\rho U_g)\rvert = \lvert \tr((\hat{\rho}_j - \rho)U_g)\rvert \leq \lVert \hat{\rho}_j - \rho\rVert_{\tr} < 0.01.\] 
Note that on a state $\sigma$ the Hadamard test accepts (i.e. outputs measurement bit $0$) with probability 
\[\frac{1}{2} + \frac{1}{2}\text{Re}(\tr(\sigma U_g)).\]
Then conditional on $E$, 
\[ \left\lvert \Pr[\text{accept on } \hat{\rho}_j] - \Pr[\text{accept on } \rho] \right\rvert < 0.005,\]
and so if $g\in H$ then a Hadamard test on each $\hat{\rho}_j$ rejects with probability at most $0.005$, while if $g\notin H$ then it accepts with probability at most $0.505$. $E$ fails with probability at most $0.02$, so when $g\in H$ the protocol output is $0$ with probability at most $2(0.005) + 0.02 < 1/3$, and when $g \notin H$ the output is $1$ with probability at most $(0.505)^2 + 0.02 < 1/3$. On every input, the error probability is therefore at most $1/3$.

$\rho$ is a mixed state of rank $|G|/|H|$ and so using the remote state preparation protocol of Theorem~\ref{thm:rspprotocol} requires \[\log \frac{|G|}{|G|/|H|} + O(\log(1/\varepsilon)) = \log |H| + O(\log(1/\varepsilon)) \leq \log k + O(\log(1/\varepsilon)) \]
bits of communication from Alice to Bob each time. The protocol calls it twice, which in total requires $2\lceil \log_2 |G| \rceil$ EPR pairs and $O(\log k)$ bits of one-way communication.
\end{proof}

\subsection{Protocols without shared entanglement}
\label{sec-unassisted-protocols}

We now describe the one-way protocol for abelian groups from Theorem~\ref{thm:abelian_subgroup}. Throughout, a character of an abelian group $G$ is a homomorphism $\chi\colon G \to \mathbb{C}^\times$, we write $\widehat{G}$ for the group of characters of $G$, and for a subgroup
$H \le G$ we write 
\[H^\perp = \{\chi \in \widehat{G} : \chi(h) = 1
\text{ for all } h \in H\}\]
for its annihilator. Note that $|H^\perp| = |G|/|H|$. 
The idea of the protocol is that $H$ is \emph{exactly} the set of elements on which every character in $H^\perp$ is trivial, so Alice can tell Bob the indices of a few randomly selected characters in $H^\perp$, and Bob tests whether they are trivial on $g$.

\begin{figure}[htbp]
\begin{tcolorbox}[colback=white,colframe=black,sharp corners,boxrule=0.5pt] 
\textbf{Given.}

\begin{itemize}
    \item A finite abelian group $G$ and a parameter $1 \leq k \leq |G|$.
\end{itemize}

Alice receives $H \in \mathcal{H}_{G,k}$ and Bob receives $g \in G$. 
Let $n = 8k$.
\medskip

\noindent\textbf{Protocol.}
\begin{enumerate}
\item Using shared randomness, Alice and Bob jointly sample characters $\chi_1, \dots, \chi_n \in \widehat{G}$ independently and uniformly at random. Alice sends Bob the first two indices $i<j\in [n]$ such that $\chi_i,\chi_j \in H^\perp$, or $\bot$ if there are fewer than two.
\item Bob outputs $1$ if $\chi_i(g) = \chi_j(g) = 1$ and he did not receive $\bot$, and $0$ otherwise. 
\end{enumerate}

\end{tcolorbox}
\caption{An efficient classical one-way protocol for $\MEM_{G,k}$ when $G$ is abelian.}
\label{fig:abelianprotocol}
\end{figure}

\begin{proof}[Proof of Theorem~\ref{thm:abelian_subgroup}]
We refer to the protocol in Figure~\ref{fig:abelianprotocol}. Since $|H^\perp| = |G|/|H|$ and there are $|\widehat{G}| = |G|$ possible characters, the probability of each independently and uniformly sampled character being in $H^\perp$ is $q = 1/|H|$.  Thus with $n = 8k \geq 8 |H|$ samples, the probability that Alice sends $\bot$ is
\[
  (1 - q)^n + nq(1 - q)^{n-1}.
\]
If $|H| = 1$, then $q = 1$ and this probability is $0$. Otherwise
$q \le 1/2$, so $(1 - q)^{-1} \le 2$. Using $1 - x \le e^{-x}$ gives
\[
  (1 - q)^n + nq(1 - q)^{n-1} \le e^{-nq} + 2nq\, e^{-nq} \le 17 e^{-8} < 0.01,
\]
where the last step uses $nq \ge 8$ and that $xe^{-x}$ is decreasing for $x\geq 1$.

If $g \in H$, then every $\chi \in H^\perp$ satisfies $\chi(g) = 1$, so Bob errs only when Alice sends $\bot$, which happens with probability less than $0.01$.

If $g \notin H$, Bob errs only when Alice does not send $\bot$ and $\chi_i(g) = \chi_j(g) = 1$. Conditioned on Alice not sending $\bot$, $\chi_i$ and $\chi_j$ are independent and uniform on $H^\perp$. Let
\[K = \{ \chi \in H^\perp: \chi(g) = 1\}.\]
Then $K$ is the annihilator of $L = \langle H, g\rangle$, so $|K| = |G|/|L|$. Since $L$ contains the disjoint sets $H$ and $gH$, we have $|L| \ge 2|H|$, and hence $|K| \leq |H^\perp|/2$. So for $\chi$ uniform on $H^\perp$, $\Pr[\chi(g) = 1] \leq 1/2$, and Bob errs with probability at most $1/4$.

On every input, the error probability is therefore at most $1/4 < 1/3$. Alice sends either two indices in $[n]$ or $\bot$, which for $n = 8k$ takes $O(\log k)$ bits of one-way communication.  
\end{proof}

We now show the remaining classical communication upper bounds in Theorem~\ref{thm:subgroup_classical_other}. Bounded-order subgroup membership is a special case of bounded-size subset membership, defined as follows.
\begin{definition}[Bounded-size subset membership]
    Let $U$ be a finite set and let $1\leq k \leq |U|$. Alice receives a subset $A \subseteq U$ with $|A|\leq k$, and Bob receives an element $x\in U$. Define the bounded-size subset membership function by
    \[ \textsc{SetMemb}_k(A,x) = \begin{cases}
        1 & \text{if } x\in A,\\
        0 & \text{if }x\notin A.
    \end{cases}\]
\end{definition}

\begin{proof}[Proof of Theorem~\ref{thm:subgroup_classical_other}]
    Every $H \in \mathcal{H}_{G,k}$ is a subset of $G$ of size at most $k$, so any protocol for $\textsc{SetMemb}_k$ with $U = G$ also solves $\MEM_{G,k}$ at the same cost. It therefore suffices to give protocols for $\textsc{SetMemb}_k$.
    In both protocols, Alice and Bob use shared randomness to pick a uniformly random hash function $h\colon U \to [3k]$. The key property is that if $x \notin A$, then by a union bound over the at most $k$ elements of $A$,
    \[
      \Pr\big[\exists\, a \in A : h(a) = h(x)\big] \le \frac{k}{3k} = \frac13 .
    \]
    \textit{Two-way protocol.} Bob sends $h(x)$ to Alice, which costs $\log k + O(1)$ bits of communication. Alice outputs $1$ if there exists some $a\in A$ such that $h(a) = h(x)$ and $0$ otherwise. If $x\in A$, this always succeeds. If $x\notin A$, then the probability that some element of $A$ hashes to the same value as $x$ is at most $1/3$ by the above bound. Note that this is really a one-way protocol in the reverse direction from Bob to Alice.

    \textit{SMP protocol.} Bob sends $h(x)$ to the referee using $\log k + O(1)$ bits. Alice sends the $3k$-bit indicator vector of $h(A) =\{h(a): a\in A\} \subseteq [3k]$ to the referee. The referee outputs $1$ if $h(x) \in h(A)$, and $0$ otherwise. The analysis is the same as in the two-way protocol.
\end{proof}

\section{Upper and lower bounds for the shifted equality problem}
\label{sec:shifteq}

Let $G$ be a finite group and let $S = \bbZ_2^r$ for some $r \geq 1$.
The communication problem which will prove Theorem \ref{thm:maintheorem} will be an instance of the following problem:
\begin{definition}[Shifted equality] \label{def:shiftedequality}
Let $G^S$ denote the set of all functions $\bg\colon S \to G$. We define the function $\ShiftEq_{G,r}: (G^S \times G^S) \times (G^S \times S) \to \{0,1\}$ by:
\[\ShiftEq_{G,r}((\bg_1,\bg_2), (\bh,s)) = \begin{cases}
    1 & \text{if }\bg_1(x)\bg_2(x+s) = \bh(x)\quad \forall x\in S,\\
    0 & \text{otherwise}.
\end{cases}\] 
\end{definition}
Going forward, we view elements of $G^S$ as arrays and omit the input $x$, for example writing $\bg_1(\shift_s(\bg_2)) = \bh$ for the success condition on inputs $\bg_1,\bg_2,\bh\in G^S$ and $s\in S$ where $\shift_s(\bg)(x) := \bg(x+s)\ \forall x \in S$.

The shifted equality problem turns out to be reducible to a specific instantiation of the bounded-order subgroup membership problem, and hence admits an efficient entanglement-assisted protocol in view of Theorem \ref{thm:upperboundsubgroupmembership}. This is summarized by the following theorem, which will be proved in Section \ref{subsec:upperboundsreduction}.
\begin{theorem} \label{thm:reductiontoMEMB}
    For any finite group $G$ and any $r \geq 1$, setting $k=2^{r+1}$, there exists a finite group $\tG$ with $|\tG| = k|G|^{k}$ along with mappings 
    \[X\colon G^S \times G^S\to \mathcal{H}_{\tG, k}, \; Y\colon G^S \times S \to \tG\]
    such that, for all Alice inputs $\bg_1, \bg_2 \in G^S$ and all Bob inputs $\bh \in G^S, s \in S$, we have:
\[\ShiftEq_{G,r}((\bg_1,\bg_2), (\bh,s)) = \MEM_{\tG, k}(X(\bg_1,\bg_2), Y(\bh, s)).\]    
\end{theorem}

On the other hand, we will show the following lower bound for one-way entanglement-unassisted quantum protocols in Section \ref{subsec:q-lb}.
\begin{theorem} \label{thm:quantumlowerbound}
    Suppose $D \geq 1$ is such that there exists $\zeta \in Z(G)$ of order 3 which is mapped to $\id$ in any representation of $G$ of dimension less than $D$. Then
    \[Q^{1, \pub}_{1/3}(\ShiftEq_{G,r}) = \Omega(\min(2^r, (\log D)/2^r )). \]
\end{theorem}

Theorem \ref{thm:maintheorem} will be derived from Theorems \ref{thm:reductiontoMEMB} and \ref{thm:quantumlowerbound} in \Cref{sec:heis-group} by providing explicit choices for $G$ and $r$.

\subsection{Upper bounds by reduction to the subgroup membership problem} \label{subsec:upperboundsreduction}

We start by establishing Theorem \ref{thm:reductiontoMEMB}. 

\begin{proof}[Proof of Theorem \ref{thm:reductiontoMEMB}]
We first construct a group ($G^{\tS}$ below) each element~$\tbg$ of which encodes a pair of elements~$\bg_1, \bg_2 \in G^S$. The underlying group~$\tG$ for the subgroup membership problem builds on~$G^{\tS}$ to include the shift~$s \in S$ in the encoding. With each pair of elements~$\bg_1, \bg_2 \in G^S$ we associate a subgroup~$H_{\tbg}$ of the group~$\tG$, where~$\tbg$ encodes the pair~$\bg_1, \bg_2$. The elements of the subgroup~$H_{\tbg}$ encode exactly the shifted products~$\bg_1 \shift_s(\bg_2)$ for~$s \in S$. We encode a pair~$\bh \in G^S, \, s \in S$ into a suitable element~$(\tbh, \ts\,) \in \tG$ to complete the reduction. Thus~$X(\bg_1, \bg_2) = H_{\tbg}$ and~$Y(\bh, s) = (\tbh, \ts\,)$. We provide the details below.

    First set
    \[\tS = S \times \bbZ_2\]
    \[\tG = G^{\tS} \rtimes_{\shift} \tS\]
    using a semidirect product, where $\shift$ denotes the homomorphism $\ts \mapsto \shift_{\ts}$ from $\tS$ to $\Aut(G^{\tS})$ given by Lemma \ref{lem:ishomomorphism}. 

    Given Alice's inputs $\bg_1, \bg_2 \in G^S$, and Bob's inputs $\bh \in G^S, s \in S$, define $\tbg, \tbh \in G^{\tS}$ by:
    \begin{align*}
    \tbg(s', b) & = \begin{cases} \bg_1(s') & \text{ if }b = 0, \\ 
                                  \bg_2(s')^{-1} & \text{ otherwise}; 
                    \end{cases}\\
    \tbh(s', b) &= \begin{cases} \bh(s') & \text{ if }b = 0, \\ 
                                 \bh(s'+s)^{-1} & \text{ otherwise}. 
                    \end{cases}  
    \end{align*}
    Thus~$\tbg$ encodes the pair~$\bg_1, \bg_2$, and~$\tbh$ encodes~$\bh$.
    
    Write $\be \in G^{\tS}$ for the array all of whose entries are~$e$, the identity of $G$. Denoting the identity element in~$\tS$ by~$0$ and the group operation in~$\tG$ by `$\cdot$', we compute, for $\ts \in \tS$:
    \begin{align*}
    (\tbg, 0) \cdot (\be, \ts\,) \cdot (\tbg, 0)^{-1} &=      (\tbg, 0) \cdot ((\be, \ts\,) \cdot  (\tbg^{-1}, 0))\\
            &= (\tbg, 0) \cdot (\shift_{\ts}(\tbg^{-1}), \ts\,)\\
            &=(\tbg \cdot \shift_{\ts}(\tbg^{-1}), \ts\,).
    \end{align*}
    Set $\ts = (s, 1)$, where~$s \in S$ is the shift offset given to Bob. We find, for $s' \in S$:
    \[(\tbg \cdot \shift_{\ts}(\tbg^{-1}))(s',0) = \tbg(s', 0) \tbg(s'+s, 1)^{-1} = \bg_1(s') \bg_2(s+s')\]
    \[(\tbg \cdot \shift_{\ts}(\tbg^{-1}))(s',1) = \tbg(s', 1) \tbg(s'+s, 0)^{-1} = \bg_2((s'+s)+s)^{-1} \bg_1(s+s')^{-1} . \]
    Thus the conjugate of~$(\be, \ts\,)$ given by~$(\tbg, 0)$ encodes the shifted product~$\bg_1 \shift_s(\bg_2)$ when~$\tbg$ encodes the pair~$\bg_1, \bg_2$.
    
    Define $H_0$ to be the following subgroup of $\tG$: 
    \[H_0 = \{(\be, \ts\,) \mid \ts \in \tS\}.\]
    The subgroup~$H_0$ has size $2^{r+1}$.
    We set:
    \begin{align*}
        X(\bg_1, \bg_2) & = (\tbg, 0)H_0(\tbg, 0)^{-1}, \qquad \text{and} \\
        Y(\bh, s) & = (\tbh, \ts \,).
    \end{align*}
    The above calculations show that $\bg_1 \shift_s(\bg_2) = \bh$ if and only if $Y(\bh, s) \in X(\bg_1, \bg_2)$, as desired. 
\end{proof}

\subsection{Lower bound without entanglement}\label{subsec:q-lb}

 We now turn to proving Theorem \ref{thm:quantumlowerbound}. For the rest of this subsection, assume that the group $G$  has an element~$\zeta \in Z(G)$ of order~$3$ (i.e., $\zeta^3 = e$ and~$\zeta \neq e$). Set $D$ to be the least dimension~$d$ of an irrep~$\rho$ of $G$ in which $\zeta$ is not trivial (i.e., $\rho(\zeta) \neq \id$). We will show the following result in Section \ref{subsubsec:finishingproof}, using an analytic theorem to be proven in Section \ref{subsubsec:preliminaryresult}.
 
 \begin{theorem} \label{thm:intermediatequantumlowerbound}
 There exists an input distribution $\mu$ such that, for all $C \geq 1$ and all even $p \geq 2$,  any one-way quantum protocol for $\ShiftEq_{G,r}$ with $C$ qubits of communication has average success probability under $\mu$ at most:
 \[\frac{1}{2} + O \!\! \left(2^{C/p}\sqrt\frac{p}{2^r} + 2^{C}\left(\frac{2^r-1}{D}\right)^{1/p}\right) . \]
 \end{theorem}

Theorem \ref{thm:intermediatequantumlowerbound} readily implies Theorem \ref{thm:quantumlowerbound}.
\begin{proof}[Proof of Theorem \ref{thm:quantumlowerbound} assuming Theorem \ref{thm:intermediatequantumlowerbound}]
Suppose there is a one-way quantum protocol for the shifted equality problem $\ShiftEq_{G,r}$ with $C$ qubits of communication under the distribution~$\mu$ over inputs, and with success probability at least~$2/3$. By \Cref{thm:intermediatequantumlowerbound}, it suffices to make both terms in its upper bound on the bias sufficiently small. Take~$p = \kappa 2^r$, even, for a sufficiently small constant~$\kappa > 0$ and~$r$ sufficiently large so that~$p \ge 2$.

Suppose first that~$C \le \kappa 2^r / 2$. Since~$p=\kappa 2^r$, we have
\[
2^{C/p} \sqrt{\frac{p}{2^r}} \le \sqrt{2\kappa}.
\]
Thus the first term can be made arbitrarily small by choosing~$\kappa$ sufficiently small.

Now let~$\kappa'>0$ be a sufficiently small constant and suppose~$C \le \kappa'(\log D)/2^r$. Using~$2^r-1\le 2^r$ and~$p=\kappa 2^r$, we have
\[
2^C\left(\frac{2^r-1}{D}\right)^{1/p} \le 2^{r/(\kappa 2^r)}2^{C-(\log D)/(\kappa 2^r)}.
\]
For~$C\ge 1$, the assumption~$C \le \kappa'(\log D)/2^r$ implies~$(\log D)/(\kappa 2^r)\ge C/(\kappa\kappa')$, and hence
\[
2^C\left(\frac{2^r-1}{D}\right)^{1/p} \le 2^{r/(\kappa 2^r)}2^{-(1/(\kappa\kappa')-1)C} \le 2^{r/(\kappa 2^r)}2^{-(1/(\kappa\kappa')-1)}.
\]
The first factor tends to~$1$ as~$r$ grows, while the second can be made arbitrarily small by choosing~$\kappa'$ sufficiently small. The case~$C=0$ is trivial.

Therefore, choosing~$\kappa$ and then~$\kappa'$ sufficiently small, with~$\kappa'\le \kappa/2$, \Cref{thm:intermediatequantumlowerbound} implies that whenever
\[
C \le \kappa' \min \{2^r,(\log D)/2^r\},
\]
the success probability of any one-way quantum protocol for~$\ShiftEq_{G,r}$ with~$C$ qubits of communication under~$\mu$ is at most~$6/10$, for sufficiently large~$r$. This contradicts the existence of such a protocol with success probability at least~$2/3$.
\end{proof}

 \subsubsection{An analytic theorem} \label{subsubsec:preliminaryresult}

The main result of this subsection will be the following estimate, which will play a key role in the proof of Theorem \ref{thm:intermediatequantumlowerbound}. 
\begin{theorem}\label{thm:ind-corr}
    Let $\{w_s\}_{s \in S}$ be a collection of functions $G^S \to \mathbb{C}$ with $\|w_s\|_\infty \leq 1$. For $\bg_1, \bg_2 \in G^S$, write:
    \[\Corr(\bg_1, \bg_2) = \prod_{s \in S} w_s(\bg_1\shift_s(\bg_2)).\]
    Also, for a phase vector $\bm{\phi} = (\phi_s)_{s \in S} \in \mathbb{F}_3^S$, write:
    \[\Corr_{\mathrm{ind}}(\bg_1, \bg_2, \bm{\phi}) = \prod_{s \in S} w_s(\zeta^{\phi_s}\bg_1\shift_s(\bg_2)).\]
    Then, if $\bg_1, \bg_2, \bm{\phi}$ are taken to be independent and uniformly random over their respective domains:
    \[\left|\bbE_{\bg_1,\bg_2}[\Corr(\bg_1,\bg_2)] - \bbE_{\bg_1,\bg_2,\bm{\phi}}[\Corr_{\mathrm{ind}}(\bg_1,\bg_2, \bm{\phi})]\right| \leq \frac{2^r-1}{D}.\]
\end{theorem}

The proof of Theorem \ref{thm:ind-corr} will be obtained by iterating the following lemma, which we prove first.
\begin{lemma}
\label{lem:G2lemma}
    Let $Q,R: G^2 \to \mathbb{C}$ be functions with~$\norm{Q}_\infty, \norm{R}_\infty \le 1$. Let $g_0,g_1,h_0,h_1$ be uniformly random elements of $G$ and $\phi$ be a uniformly random element of $\mathbb{F}_3$. We have
\begin{equation}
\label{eq-mixing-bound}
\Big| \mathbb{E} \big[Q(g_0h_0, g_1h_1)R(g_0h_1, g_1h_0) \big] - \mathbb{E} \big[Q(\zeta^\phi g_0h_0, \zeta^\phi g_1h_1)R(g_0h_1, g_1h_0) \big] \Big| \leq \frac{1}{D} \enspace.
\end{equation}
\end{lemma}
In order to prove Lemma \ref{lem:G2lemma}, we view functions~$Q : G^2 \to \complex$ as vectors in~$\complex^G \tensor \complex^G$, and represent them, in scaled form, in the bra-ket notation as possibly unnormalized states:
\[
\ket{Q} = \frac{1}{|G|} \sum_{g_1, g_2 \in G} Q(g_1, g_2) \ket{g_1 g_2}.
\]
The standard inner product of two such vectors~$Q,R$ then evaluates to an expectation over uniformly random elements as
\[
\langle Q | R \rangle = \frac{1}{|G|^2} \sum_{g_1, g_2 \in G} \overline{Q(g_1, g_2)} R(g_1, g_2),
\]
and induces the~$\ell_2$ norm
\[
\norm{ \ket{Q} }_2^2 = \frac{1}{|G|^2} \sum_{g_1, g_2 \in G} \abs{ Q(g_1, g_2)}^2.
\]

The expression in \cref{eq-mixing-bound} which \Cref{lem:G2lemma} bounds may be written in terms of a quadratic form involving the two operators~$K,P$ that we define below. The bound may then be derived by estimating the operator norm of~$(\id - P)K$. It turns out that the operator~$(\id - P)K$ has a simple form when viewed in the Fourier basis. Its operator norm may then be inferred readily. We develop these ideas next.

Recall from \Cref{sec-rep-theory} that~$\hG$ denotes a complete set of nonisomorphic unitary irreps of~$G$, and~$d_\rho$ denotes the dimension of a representation~$\rho$. Recall that~$D$ is the least integer~$d$ such that there is an irrep~$\rho$ of $G$ of dimension~$d$ in which $\zeta$ is not trivial:
\begin{equation}
\label{eq-min-irrep-dim}
D = \min \left\{ d_\rho : \rho \in \hG, \; \rho(\zeta) \neq \id \right\}.
\end{equation}
Also recall that the Fourier Transform~$F_G$ gives an isomorphism between~$\complex^G$ and~$V(\hG)$, and is unitary. We omit the subscript~$G$ from~$\qft_G$ for ease of notation.

For a function~$Q : G^2 \to \complex$, let~$\ket{\hQ} = (\qft \tensor \qft) \ket{Q}$ denote its Fourier transform. By unitarity of~$\qft$, we get the Parseval Identity
\[
\norm{ \ket{\hQ}}_2 = \norm{ \ket{Q}}_2 
\]
in the~$\ell_2$ norm induced by the standard inner product in the respective spaces. Similarly, for a linear operator~$T$ on~$\complex^G \tensor \complex^G$, let~$\hT = (F \tensor F) T (F^\adjoint \tensor F^\adjoint)$ be an operator on~$V(\hG) \tensor V(\hG)$. We have
\begin{align}
\nonumber
\norm{\hT}_\infty & = \max \left\{ \norm{\hT \ket{\hQ}}_2 : \ket{\hQ} \in V(\hG) \tensor V(\hG), ~ \norm{ \ket{\hQ}}_2 = 1 \right\} \\
\nonumber
    & = \max \left\{ \norm{T \ket{Q}}_2 : Q \in \complex^G \tensor \complex^G, ~ \norm{ \ket{Q}}_2 = 1 \right\} \\
\label{eq-operator-parseval}
    & = \norm{T}_\infty.
\end{align}

We define two linear operators~$K,P$ on~$\complex^G \tensor \complex^G$ mentioned above. Let~$K$ be defined by
\[
K = \frac{1}{|G|} \sum_{u,v,w \in G} \ket{u,v} \! \bra{uw, v w^{-1}},
\]
and~$P$ be defined by
\[
P = \frac{1}{3} \sum_{u,v \in G} \sum_{j \in \finitefield_3} \ket{u,v} \! \bra{ \zeta^j u, \zeta^j v}.
\]
Note that~$K$ and~$P$ are both Hermitian. We have
\begin{align}
\nonumber
K^\adjoint & = \frac{1}{|G|} \sum_{u,v,w \in G} \ket{uw, v w^{-1}} \!\! \bra{u,v} \\
\label{eq-k-hermitian}
  & = \frac{1}{|G|} \sum_{u',v',w \in G} \ket{u', v'} \!\! \bra{u' w^{-1}, v' w} = K,
\end{align}
and similarly, $P^\adjoint = P$. We derive the following bound on operator norm in anticipation of \Cref{lem:G2lemma}.

\begin{lemma}
\label{lem-mixing-bound}
$\norm{(\id - P)  K}_\infty \le \tfrac{1}{D} \,$.
\end{lemma}
\begin{proof}
Let~$\hP = (\qft \tensor \qft) P (\qft^\adjoint \tensor \qft^\adjoint)$ and~$\hK = (\qft \tensor \qft) K (\qft^\adjoint \tensor \qft^\adjoint) $. By \cref{eq-operator-parseval}, we need only bound~$\norm{(\id - \hP)  \hK}_\infty$.

Since~$\zeta \in Z(G)$, for any irrep~$\rho$, $\rho(\zeta)$ commutes with~$\rho(g)$ for all~$g \in G$. As~$\zeta$ has order~$3$, the Schur Lemma (\Cref{prop-schur}) gives~$\rho(\zeta) = \upomega^{s(\rho)} \id$ for some~$s(\rho) \in \finitefield_3$, where~$\upomega = \exp(2\pi \complexi / 3)$.

By the definition of the Fourier transform~$\qft$, the entries of~$\hP$ and~$\hK$ are indexed by quadruples of irreps of~$G$ and their row and column indices. For irreps~$\rho_k$ and indices~$i_k, j_k \in [d_{\rho_k}]$, where~$k \in [4]$, by the orthogonality relations between irreps (\Cref{prop-irrep-orthogonality}) we have
\begin{align*}
\lefteqn{ \bra{\rho_1, i_1, j_1} \bra{\rho_2, i_2, j_2} \hP \ket{\rho_3, i_3, j_3} \ket{\rho_4, i_4, j_4} } \\
  & = \frac{1}{3 |G|^2} \sum_{u,v \in G; j \in \finitefield_3} \sqrt{d_{\rho_1} d_{\rho_2}} \, \rho_1(u)(i_1, j_1) \, \rho_2(v)(i_2, j_2) \\
  & \qquad \mbox{} \times  \sqrt{d_{\rho_3} d_{\rho_4}} \, (\rho_3( \zeta^j u)(i_3, j_3))^* (\rho_4(\zeta^j v)(i_4, j_4))^* \\
  & = \sum_{l_3 \in [d_{\rho_3}]} \left[ \frac{1}{|G|} \sum_{u \in G} \sqrt{d_{\rho_1} d_{\rho_3}} \, \rho_1(u)(i_1, j_1) \, (\rho_3(u)(l_3, j_3))^* \right] \\
  & \qquad \mbox{} \times \sum_{l_4 \in [d_{\rho_4}]} \left[ \frac{1}{|G|} \sum_{v \in G} \sqrt{d_{\rho_2} d_{\rho_4}} \, \rho_2(v)(i_2, j_2) (\rho_4(v)(l_4, j_4))^* \right] \\
  & \qquad \mbox{} \times \left[ \frac{1}{3}  \sum_{j \in \finitefield_3} (\rho_3(\zeta^j)(i_3, l_3))^* (\rho_4(\zeta^j)(i_4, l_4))^* \right] \\
  & = \sum_{l_3 \in [d_{\rho_3}]} \sum_{l_4 \in [d_{\rho_4}]} \updelta_{\rho_1 \rho_3} \updelta_{i_1 l_3} \updelta_{j_1 j_3} \times \updelta_{\rho_2 \rho_4} \updelta_{i_2 l_4} \updelta_{j_2 j_4} \times \frac{1}{3}  \sum_{j \in \finitefield_3} \upomega^{-j (s(\rho_3) + s(\rho_4))} \updelta_{i_3 l_3} \updelta_{i_4 l_4} \enspace,
\end{align*}
using the form of the operators~$\rho(\zeta)$ for irreps~$\rho$ stated above.
Simplifying this expression, we get
\[
\hP = \sum_{\rho_1, \rho_2 \,:\, s(\rho_1) + s(\rho_2) = 0 \! \! \pmod{3}} \density{\rho_1} \tensor \id_{d_{\rho_1}}^{\tensor 2} \tensor \density{\rho_2} \tensor \id_{d_{\rho_2}}^{\tensor 2} \enspace.
\]

We similarly derive an explicit expression for~$\hK$. For irreps~$\rho_k$ and indices~$i_k, j_k \in [d_{\rho_k}]$, where~$k \in [4]$, by the orthogonality relations of irreps we have
\begin{align*}
\lefteqn{ \bra{\rho_1, i_1, j_1} \bra{\rho_2, i_2, j_2} \hK \ket{\rho_3, i_3, j_3} \ket{\rho_4, i_4, j_4} } \\
  & = \frac{1}{|G|^3} \sum_{u,v,w \in G} \sqrt{d_{\rho_1} d_{\rho_2}} \, \rho_1(u)(i_1, j_1) \, \rho_2(v)(i_2, j_2) \\
  & \qquad \mbox{} \times  \sqrt{d_{\rho_3} d_{\rho_4}} \, (\rho_3(uw)(i_3, j_3))^* (\rho_4(v w^{-1})(i_4, j_4))^* \\
  & = \sum_{l_3 \in [d_{\rho_3}]} \left[ \frac{1}{|G|} \sum_{u \in G} \sqrt{d_{\rho_1} d_{\rho_3}} \, \rho_1(u)(i_1, j_1) \, (\rho_3(u)(i_3, l_3))^* \right] \\
  & \qquad \mbox{} \times \sum_{l_4 \in [d_{\rho_4}]} \left[ \frac{1}{|G|} \sum_{v \in G} \sqrt{d_{\rho_2} d_{\rho_4}} \, \rho_2(v)(i_2, j_2) (\rho_4(v)(i_4, l_4))^* \right] \\
  & \qquad \mbox{} \times \left[ \frac{1}{|G|}  \sum_{w \in G} (\rho_3(w)(l_3, j_3))^* (\rho_4(w^{-1})(l_4, j_4))^* \right] \\
  & = \sum_{l_3 \in [d_{\rho_3}]} \sum_{l_4 \in [d_{\rho_4}]} \updelta_{\rho_1 \rho_3} \updelta_{i_1 i_3} \updelta_{j_1 l_3} \times \updelta_{\rho_2 \rho_4} \updelta_{i_2 i_4} \updelta_{j_2 l_4} \times \frac{1}{d_{\rho_3}} \updelta_{\rho_3 \rho_4} \updelta_{l_3 j_4} \updelta_{j_3 l_4} \enspace,
\end{align*}
as~$\rho_4( w^{-1})( l_4, j_4) = (\rho_4( w))^\adjoint ( l_4, j_4) = (\rho_4( w)( j_4, l_4))^*$.
Simplifying this expression, we get
\[
\hK \quad = \quad \sum_{\rho \in \hG} ~ \sum_{j_1, j_2 \in [d_{\rho}]} \frac{1}{d_{\rho}} \density{\rho} \tensor \id_{d_{\rho}} \tensor \ketbra{j_1}{j_2} \tensor \density{\rho} \tensor \id_{d_{\rho}} \tensor \ketbra{j_2}{j_1}.
\]

The expressions for~$\hP$ and~$\hK$ give us
\begin{align*}
(\id - \hP) \hK & \quad = \quad \hK - \sum_{\rho \,:\, 2 s(\rho) = 0} \frac{1}{d_{\rho}}  \sum_{j_1, j_2 \in [d_{\rho}]} \density{\rho} \tensor \id_{d_{\rho}} \tensor \ketbra{j_1}{j_2} \tensor \density{\rho} \tensor \id_{d_{\rho}} \tensor \ketbra{j_2}{j_1} \\
  & \quad = \quad \sum_{\rho \,:\, \rho(\zeta) \neq \id} \frac{1}{d_{\rho}} \sum_{j_1, j_2 \in [d_{\rho}]} \density{\rho} \tensor \id_{d_{\rho}} \tensor \ketbra{j_1}{j_2} \tensor \density{\rho} \tensor \id_{d_{\rho}} \tensor \ketbra{j_2}{j_1}.
\end{align*}
This operator has a block-diagonal form with blocks indexed by $\rho$, where the $\rho$-block is $\frac{1}{d_\rho}$ times a permutation operator. Thus
\[ \left\|(I-\widehat P)\widehat K\right\|_\infty = \max_{\rho:\rho(\zeta)\neq I}\frac{1}{d_\rho} = \frac{1}{D}, \]
as claimed.
\end{proof}

This lets us complete the proof of Lemma \ref{lem:G2lemma}. 
\begin{proof}[Proof of Lemma \ref{lem:G2lemma}]
We express the left hand side of \cref{eq-mixing-bound} as a quadratic form involving the operators~$K,P$. As a result it suffices to bound the norm of an expression in these operators to conclude \cref{eq-mixing-bound}.

Let~$x = g_0 h_0$, $y = g_1 h_1$, and~$z = h_0^{-1} h_1$. Then~$x,y,z$ are independent and uniformly distributed over~$G$. Defining~$\overline{Q}$ as the point-wise complex conjugate of~$Q$, we see that
\begin{align*}
\expct \big[ Q( g_0h_0, g_1h_1) R( g_0h_1, g_1h_0) \big] = \expct \big[ Q( x, y) R( x z, y z^{-1} ) \big]
    = \left. \! \left\langle \overline{Q} \right| KR \right\rangle,
\end{align*}
and
\begin{align*}
\expct \big[ Q( \zeta^\phi g_0h_0, \zeta^\phi g_1h_1) R( g_0h_1, g_1h_0) \big] = \expct \big[ Q( \zeta^\phi x, \zeta^\phi y) R( x z, y z^{-1} ) \big]
    = \left. \! \left\langle P \overline{Q} \right| KR \right\rangle.
\end{align*}
So their difference may be bounded as
\begin{align*}
\lefteqn{ \Big| \expct \big[ Q( g_0h_0, g_1h_1) R( g_0h_1, g_1h_0) \big]  - \expct \big[ Q( \zeta^\phi g_0h_0, \zeta^\phi g_1h_1) R( g_0h_1, g_1h_0) \big] \Big| } \\
  & = \left. |\! \left\langle (\id - P) \overline{Q} \right| KR \right\rangle |\\
  & = \left. |\! \left\langle \overline{Q} \right| (\id - P)  KR \right\rangle | & (\text{as } P \text{ is Hermitian}) \\
  & \le \norm{\ket{\overline{Q}}}_2 \norm{(\id - P)  K}_\infty \norm{\ket{R}}_2 & (\text{by Cauchy-Schwarz}) \\
  & \le \frac{1}{D} \enspace, & (\text{\Cref{lem-mixing-bound}})
\end{align*}
as~$\norm{Q}_\infty, \norm{R}_\infty \le 1$ imply that the (scaled) $\ell_2$-norms~$\norm{\ket{\overline{Q}}}_2, \norm{\ket{R}}_2 \le 1$. This proves the lemma.
\end{proof}

With the lemma in hand, we complete the proof of Theorem \ref{thm:ind-corr}. We start with the following lemma:
\begin{lemma} \label{lem:Misinvertible}
For~$x,y \in S = \bbZ_2^r$, let~$x \cdot y = \sum_{i=1}^r x_i y_i \pmod{2}$.
    Let $M$ be the $2^r \times 2^r$ matrix over $\mathbb{F}_3$ whose rows and columns are indexed by $x,y \in S$, with entries
    \[M_{x,y} = \begin{cases}1 & \text{if } x \cdot y = 0,\\ 0 & \text{if } x \cdot y = 1.\end{cases}\]
    Then $M$ is invertible.
\end{lemma}
\begin{proof}
    Observe that over $\mathbb{F}_3$, we have~$M_{x,y} = -(1 + (-1)^{x \cdot y})$ for all~$x,y \in S$.
    Hence, we have~$M = -(\mathsf{J} + \sfH^{\otimes r})$,
    where $\mathsf{J}$ is the all-ones matrix and $\sfH = \begin{pmatrix}1&1\\1&-1\end{pmatrix}$ is the unnormalized Hadamard matrix. The determinant of~$\sfH$ is $-2$, which is nonzero in $\mathbb{F}_3$. So $\sfH$ and $\sfH^{\otimes r}$ are invertible over $\mathbb{F}_3$.

    For each $n \geq 1$, we write $\ket{1^n} = (1,\ldots,1)^T \in \mathbb{F}_3^n$ for the all-ones column vector and $\bra{1^n}$ for its transpose. Thus $\mathsf{J} = \density{1^{2^r}}$. Lemma \ref{lem:conditionforinvertibility} says that $-M$ (and hence $M$) is invertible if and only if $-1 \neq \bra{1^{2^r}}(\sfH^{\otimes r})^{-1}\ket{1^{2^r}}$. Since $\sfH^{-1} = 2^{-1}\sfH$, we have
    \[\bra{1^2}\sfH^{-1}\ket{1^2} = 2^{-1}(1+1+1-1) = 1.\]
    Using $\ket{1^{2^r}} = \ket{1^2}^{\otimes r}$, we therefore obtain
    \[\bra{1^{2^r}}(\sfH^{\otimes r})^{-1}\ket{1^{2^r}} = \bigl(\bra{1^2}\sfH^{-1}\ket{1^2}\bigr)^r = 1 \neq -1\]
    over $\mathbb{F}_3$, proving the claim.
\end{proof}

\begin{proof}[Proof of Theorem~\ref{thm:ind-corr}]
    We fix an enumeration $t_1,\ldots,t_{2^r}$ of $S$ with $t_1=0^r$. We choose $\bm{\phi}'=(\phi'_s)_{s\in S}\in\mathbb{F}_3^S$ uniformly at random and independently of $\bg_1,\bg_2$. For $k\in[2^r]$ and $s\in S$, we set
    \[b_{k,s}=\sum_{j=1}^k\phi'_{t_j}M_{s,t_j}\]
    over $\mathbb{F}_3$, where $M$ is the matrix from Lemma~\ref{lem:Misinvertible}. We define
    \[\Corr_k(\bg_1,\bg_2,\bm{\phi}')=\prod_{s\in S}w_s\bigl(\zeta^{b_{k,s}}\bg_1\shift_s(\bg_2)\bigr).\]
    Since $M_{s,t_1}=1$ for every $s$, we have $b_{1,s}=\phi'_{t_1}$. For each fixed value of~$\phi'_{t_1}$, multiplication of every coordinate of $\bg_1$ by $\zeta^{\phi'_{t_1}}$ preserves the uniform distribution of~$\bg_1$. Consequently,
    \[\mathbb{E}[\Corr_1(\bg_1,\bg_2,\bm{\phi}')]=\mathbb{E}[\Corr(\bg_1,\bg_2)].\]
    At the other endpoint, $b_{2^r,s}=(M\bm{\phi}')_s$. Since $M$ is invertible by the previous lemma, $M\bm{\phi}'$ is uniformly random in $\mathbb{F}_3^S$ and remains independent of $\bg_1,\bg_2$. Thus
    \[\mathbb{E}[\Corr_{2^r}(\bg_1,\bg_2,\bm{\phi}')]=\mathbb{E}[\Corr_{\mathrm{ind}}(\bg_1,\bg_2,\bm{\phi})].\]
    We will prove that, for every $1\leq k\leq 2^r-1$,
    \begin{equation}\label{eq:phase-increment}
    \left|\mathbb{E}[\Corr_{k+1}(\bg_1,\bg_2,\bm{\phi}')]-\mathbb{E}[\Corr_k(\bg_1,\bg_2,\bm{\phi}')]\right|\leq\frac1D.
    \end{equation}
    To compare these two quantities, we first condition on the ``phases''~$\phi'_{t_j}$ occurring in the sum~$b_{k,s}$. This makes their ``rotations'' (the product with~$\zeta^{b_{k,s}}$) part of fixed functions of the shifted products~$\bg_1 \shift_s(\bg_2)$. The additional phase~$\phi'_{t_{k+1}}$ at step $k+1$ can then be treated using Lemma~\ref{lem:G2lemma}.

    We fix $k$ and condition on $\phi'_{t_1},\ldots,\phi'_{t_k}$. These values determine $b_{k,s}$ for every $s$. For these fixed values, we define
    \[\tw_s(\bh)=w_s(\zeta^{b_{k,s}}\bh),\qquad \bh\in G^S.\]
    Each $\tw_s$ is bounded in absolute value by $1$. The next phase $\phi'_{t_{k+1}}$ is still uniform and independent of $\bg_1,\bg_2$. We define $\ell:S\to\mathbb{Z}/2\mathbb{Z}$ by $\ell(s)=t_{k+1}\cdot s$, with the dot product computed modulo $2$. This linear form is nonzero because $t_{k+1}\ne 0^r$. By the definition of $M$,
    \[b_{k+1,s}=\begin{cases}b_{k,s}+\phi'_{t_{k+1}}&\text{if }\ell(s)=0,\\ b_{k,s}&\text{if }\ell(s)=1.\end{cases}\]
    Hence the new phase~$\phi'_{t_{k+1}}$ rotates exactly the arguments of the functions~$\tw_s$ whose shifts~$s$ satisfy $\ell(s)=0$.

    We choose $g_0,g_1,h_0,h_1\in G$ uniformly and independently of each other and of all previously sampled random variables. For each coordinate $x\in S$, we define
    \[\bu(x)=g_{\ell(x)},\qquad \bv(x)=h_{\ell(x)},\]
    and set
    \[\tbg_1(x)=\bg_1(x)\bu(x)^{-1},\qquad \tbg_2(x)=\bv(x)^{-1}\bg_2(x).\]
    For every fixed choice of the four auxiliary group elements~$g_0, g_1, h_0, h_1$, this is a bijective change of variables from $(\bg_1,\bg_2)$ to $(\tbg_1,\tbg_2)$. The conditional distribution of the modified arrays~$(\tbg_1,\tbg_2)$ is therefore uniform on $G^S\times G^S$ and does not depend on~$g_0, g_1, h_0, h_1$. In particular, the modified arrays~$(\tbg_1,\tbg_2)$, the four auxiliary elements~$g_0, g_1, h_0, h_1$, and the next phase~$\phi'_{t_{k+1}}$ are mutually independent, even after conditioning on the earlier phases.

    We now also condition on $\tbg_1,\tbg_2$. Since $\ell(x+s)=\ell(x)+\ell(s)$, with addition modulo $2$, the original shifted product has coordinates
    \begin{equation}\label{eq:shifted-coordinate}
    \bigl(\bg_1\shift_s(\bg_2)\bigr)(x)=\tbg_1(x)g_{\ell(x)}h_{\ell(x)+\ell(s)}\tbg_2(x+s).
    \end{equation}
    For $\ell(s)=0$, the middle product is either $g_0h_0$ or $g_1h_1$. For $\ell(s)=1$, it is either $g_0h_1$ or $g_1h_0$. To collect the two sets of shifts explicitly, we define, for $a_0,a_1\in G$,
    \begin{align*}
    R_1(a_0,a_1)&=\prod_{\substack{s\in S \,:\\\ell(s)=0}}\tw_s\!\left(\bigl(\tbg_1(x)a_{\ell(x)}\tbg_2(x+s)\bigr)_{x\in S}\right),\\
    R_2(a_0,a_1)&=\prod_{\substack{s\in S \,:\\\ell(s)=1}}\tw_s\!\left(\bigl(\tbg_1(x)a_{\ell(x)}\tbg_2(x+s)\bigr)_{x\in S}\right).
    \end{align*}
    Each parenthesized tuple is an array in $G^S$. The functions~$R_1, R_2$ are fixed under the conditioning, and $\|R_1\|_\infty,\|R_2\|_\infty\leq1$, since every factor is bounded by $1$ in absolute value.

    Equation~\eqref{eq:shifted-coordinate} gives the identity
    \[\Corr_k(\bg_1,\bg_2,\bm{\phi}')=R_1(g_0h_0,g_1h_1)R_2(g_0h_1,g_1h_0).\]
    Moreover, $\zeta$ commutes with every element of $G$. Multiplying both arguments of $R_1$ by $\zeta^{\phi'_{t_{k+1}}}$ therefore multiplies every coordinate of every array appearing in $R_1$ by that same element. This is exactly the additional rotation for shifts with $\ell(s)=0$; the factors in $R_2$ are unchanged. Consequently,
    \[\Corr_{k+1}(\bg_1,\bg_2,\bm{\phi}')=R_1(\zeta^{\phi'_{t_{k+1}}}g_0h_0,\zeta^{\phi'_{t_{k+1}}}g_1h_1)R_2(g_0h_1,g_1h_0).\]
    Under the conditioning, the four auxiliary group elements and $\phi'_{t_{k+1}}$ remain independent and uniform. Lemma~\ref{lem:G2lemma} therefore bounds the absolute difference between the conditional expectations of these last two displayed quantities by $1/D$. This bound holds for every fixed choice of the earlier phases and modified arrays. Averaging over those choices and using the triangle inequality proves~\eqref{eq:phase-increment}.

    Finally, the two endpoint identities and~\eqref{eq:phase-increment} imply
    \begin{align*}
    \lefteqn{ \Big| \mathbb E \big[\Corr(\bg_1,\bg_2) \big] - \mathbb E \big[ \Corr_{\mathrm{ind}}(\bg_1,\bg_2,\bm{\phi}) \big] \Big| } \\
    & \leq\sum_{k=1}^{2^r-1} \Big| \mathbb E \big[ \Corr_{k+1}(\bg_1,\bg_2,\bm{\phi}') \big] - \mathbb E \big[ \Corr_k(\bg_1,\bg_2,\bm{\phi}') \big] \Big|\leq\frac{2^r-1}{D},
    \end{align*}
    as claimed.
\end{proof}

\subsubsection{Finishing the proof of Theorem \ref{thm:intermediatequantumlowerbound}} \label{subsubsec:finishingproof}

 We now turn to completing the proof of Theorem \ref{thm:intermediatequantumlowerbound}. We will need the following theorem.
\begin{theorem}\label{thm:estimateforT}
    Take $d \geq 1$, and let $\{O_s\}_{s \in S}$ be a collection of maps $G^S \to \Herm(d)$ satisfying, for all $\bh \in G^S$, given a uniformly random $\phi \in \mathbb{F}_3$:
    \[
\|O_s(\bh)\|_\infty\leq1,
\qquad
\bbE_\phi[O_s(\zeta^\phi \bh)] = 0.
\]
    Define the map $T: G^S \times G^S \to \Herm(d)$ as:
    \[T(\bg_1, \bg_2) = \bbE_{s \in S}[O_s(\bg_1 \shift_s(\bg_2))]\]
    Then, for all even $p \geq 2$, 
    \[\bbE_{\bg_1, \bg_2}[\tr(T(\bg_1, \bg_2)^p)] \leq O\left(d \left(\frac{4p}{2^r}\right)^{p/2} + d^p\left(\frac{2^r-1}{D}\right)\right)\]
\end{theorem}
\begin{proof}
    For a given $\bm{\phi} \in \mathbb{F}_3^S$, write:
    \[\tT(\bg_1, \bg_2, \bm{\phi}) = \bbE_{s \in S}[O_s(\zeta^{\phi_s} \bg_1 \shift_s(\bg_2))]\]
    We will show that, if $\bg_1, \bg_2, \bm{\phi}$ are uniformly random, we have:
\begin{equation} \label{eq:bound-tildeT} 
    \bbE_{\bg_1, \bg_2, \bm{\phi}}[\tr(\tT(\bg_1, \bg_2, \bm{\phi})^p)] \leq O\left(d \left(\frac{4p}{2^r}\right)^{p/2}\right)
    \end{equation}
    \begin{equation} \label{eq:bound-diff} 
    \left|\bbE_{\bg_1, \bg_2}[\tr(T(\bg_1, \bg_2)^p)] - \bbE_{\bg_1, \bg_2, \bm{\phi}}[\tr(\tT(\bg_1, \bg_2, \bm{\phi})^p)]\right| \leq O\left(d^p\left(\frac{2^r-1}{D}\right)\right)
    \end{equation}
    These two equations can be seen to imply the result.

    To derive equation \eqref{eq:bound-tildeT}, note that for a given choice of $\bg_1, \bg_2$, the matrices $O_s(\zeta^{\phi_s} \bg_1 \shift_s(\bg_2))$ are all independent, and the hypotheses guarantee that they have spectral norm bounded by one and have expectation zero. The claimed bound then follows from Corollary \ref{cor:ind-moment} for every choice of $\bg_1, \bg_2$, and hence also for the average over them. 
   We now turn to deriving equation \eqref{eq:bound-diff}. To do this, we note that we can expand the trace of the product of $p$ matrices as
\[ \tr(A_1\ldots A_p) = \sum_{i_1,\ldots, i_p} (A_1)_{i_1i_2}(A_2)_{i_2i_3}\ldots (A_p)_{i_pi_1}. \]
Using this, we can expand the trace of the $p$-th power of $T$ as:
\begin{equation}\label{eq:Tp-expansion}
\tr\left(T(\bg_1,\bg_2)^p\right)= \frac{1}{|S|^p}
\sum_{s_1,\ldots,s_p\in S}
\sum_{i_1,\ldots,i_p=1}^{d}
\prod_{j=1}^p
\left[O_{s_j}\bigl(\bg_1\shift_{s_j}(\bg_2)\bigr)\right]_{i_j,i_{j+1}}.
\end{equation}
where $\left[O_{s_j}\bigl(\bg_1\shift_{s_j}(\bg_2)\bigr)\right]_{i_j,i_{j+1}}$ denotes the $(i_j,i_{j+1})$-th entry of the $d\times d$ matrix, and we have adopted the convention $i_{p+1}=i_1$. Similarly,
\begin{equation}\label{eq:tTp-expansion}
\tr\left(\tT(\bg_1,\bg_2, \bm{\phi})^p\right)= \frac{1}{|S|^p}
\sum_{s_1,\ldots,s_p\in S}
\sum_{i_1,\ldots,i_p=1}^{d}
\prod_{j=1}^p
\left[O_{s_j}\bigl(\zeta^{\phi_{s_j}}\bg_1\shift_{s_j}(\bg_2)\bigr)\right]_{i_j,i_{j+1}}.
\end{equation}

Fixing $\bs=(s_1,\ldots,s_p)$ and $\bi=(i_1,\ldots,i_p)$, we will now define a function that is the product of all the terms in one summand in \eqref{eq:Tp-expansion} that have their $s_j=s$, for some $s\in S$. Define
\[
w^{\bs,\bi}_s(\bh)=
\prod_{\substack{1\leq j\leq p\\s_j=s}}
[O_{s_j}(\bh)]_{i_j,i_{j+1}},
\]
where an empty product is $1$. \eqref{eq:Tp-expansion} and \eqref{eq:tTp-expansion} can now be rewritten in terms of these functions as
\begin{gather}
\tr\left(T(\bg_1,\bg_2)^p\right) = \frac{1}{|S|^p}  \sum_{\bs,\bi}\prod_{s\in S}w^{\bs,\bi}_s\bigl(\bg_1\shift_s(\bg_2)\bigr) \label{eq:Tp-ws} \\
\tr\left(\tT(\bg_1,\bg_2, \bm{\phi})^p\right) = \frac{1}{|S|^p}  \sum_{\bs,\bi}\prod_{s\in S}w^{\bs,\bi}_s\bigl(\zeta^{\phi_s}\bg_1\shift_s(\bg_2)\bigr). \label{eq:tTp-ws}
\end{gather}

We note that every matrix entry of the $O_s(\bh)$ matrix has absolute value at most $\|O_s(\bh)\|_\infty\leq 1$, so $|w^{\bs,\bi}_s(\bh)|\leq 1$, and we can apply Theorem~\ref{thm:ind-corr} to these functions. The summands in \eqref{eq:Tp-ws} and \eqref{eq:tTp-ws} are precisely the two expressions that Theorem~\ref{thm:ind-corr} compares, in expectation over $\bg_1,\bg_2$. The theorem gives us
\[ \left|\bbE_{\bg_1,\bg_2}\left[\prod_{s\in S}w^{\bs,\bi}_s\bigl(\bg_1\shift_s(\bg_2)\bigr)\right] - \bbE_{\bg_1,\bg_2,\bm{\phi}}\left[\prod_{s\in S}w^{\bs,\bi}_s\bigl(\zeta^{\phi_s}\bg_1\shift_s(\bg_2)\bigr)\right]\right| \leq \frac{2^r-1}{D}.\]
Thus, using the above equation and equations~\eqref{eq:Tp-ws} and \eqref{eq:tTp-ws} we get,
\begin{align*}
& \left|
\bbE_{\bg_1,\bg_2}\left[\tr\left(T(\bg_1,\bg_2)^p\right)\right]
-
\bbE_{\bg_1,\bg_2,\bm{\phi}}\left[\tr\left(\tT(\bg_1,\bg_2,\bm{\phi})^p\right)\right]
\right| \\
&\leq \frac{1}{|S|^p}  \sum_{\bs,\bi}\left|\bbE_{\bg_1,\bg_2}\left[\prod_{s\in S}w^{\bs,\bi}_s\bigl(\bg_1\shift_s(\bg_2)\bigr)\right] - \bbE_{\bg_1,\bg_2,\bm{\phi}}\left[\prod_{s\in S}w^{\bs,\bi}_s\bigl(\zeta^{\phi_s}\bg_1\shift_s(\bg_2)\bigr)\right]\right| \\
& \leq \frac{1}{|S|^p}  \sum_{\bs,\bi}\frac{2^r-1}{D} \\ 
&\leq \frac{2^r-1}{D}d^p 
\end{align*}
where the fact that there are $d^p$ choices for $\bi$ was used in the last line. This completes the derivation of equation \eqref{eq:bound-diff}, and hence finishes the proof.
\end{proof}

We can now complete the proof of Theorem \ref{thm:intermediatequantumlowerbound}.
\begin{proof}[Proof of Theorem \ref{thm:intermediatequantumlowerbound}]
    We take the promised input distribution $\mu$ to be the following. First, Alice's inputs $\bg_1, \bg_2$ are sampled uniformly in $G^S$, and $s$ is sampled uniformly in $S$. Set $\bh_0=  \bg_1\shift_s(\bg_2)$. Then, taking $\phi\in\bbF_3$ to be sampled with the following probabilities:
\[ \Pr[\phi=0] = \frac12, \qquad \Pr[\phi=1]=\Pr[\phi=2]=\frac14,\]
Bob's inputs are taken to be $s$ and $\bh=\zeta^\phi \bh_0$. Note that since $\zeta$ has order 3 by assumption, we have that $\ShiftEq_{G,r}((\bg_1, \bg_2), 
(\bh, s)) = 1$ if and only if $\phi = 0$. 

We now use the characterization of the optimal success probability of a quantum protocol which is given by Lemma \ref{lem:op-norm-bound}. Letting $\{M^{\bh, s}_b\}$ be a collection of $2^C$-dimensional binary POVMs indexed by Bob's inputs, we see that given Alice's inputs $\bg_1, \bg_2$, setting $\bh_0 = \bg_1(\shift_s(\bg_2))$, we have:
\begin{align*}
\bbE_{\phi}[M^{\bh, s}_{\ShiftEq_{G,r}((\bg_1, \bg_2), 
(\bh, s))}]  &= \frac{1}{2}M^{\bh_0, s}_1 + \frac{1}{4}M^{\zeta \bh_0, s}_0 + \frac{1}{4}M^{\zeta^2 \bh_0, s}_0    \\
            &= \frac\id2 +  \frac{1}{2}M^{\bh_0, s}_1 - \frac{1}{4}M^{\zeta \bh_0, s}_1 - \frac{1}{4}M^{\zeta^2 \bh_0, s}_1
\end{align*}
where we used that $M^{\bh, s}_0 + M^{\bh, s}_1 = \id$ since these are POVMs. Hence, if we set:
\[O_s(\bh_0)=\frac12 M^{\bh_0,s}_1-\frac14 M^{\zeta\bh_0,s}_1-\frac14 M^{\zeta^2\bh_0,s}_1\]
We see from Lemma \ref{lem:op-norm-bound} that, in order to bound the success probability of any protocol, we need to bound:
\[1/2 + \bbE_{\bg_1, \bg_2}[\|\bbE_s[O_s(\bh_0)]\|_\infty]\]
We will show that the functions $O_s$ satisfy the conditions of Theorem \ref{thm:estimateforT}, in which we take $d=2^C$. This then implies, for any even $p \geq 2$:
\begin{align*}
    \bbE_{\bg_1, \bg_2}[\|\bbE_s[O_s(\bh_0)]\|_\infty] &\leq \bbE_{\bg_1, \bg_2}[\tr(\bbE_s[O_s(\bh_0)]^p)^{1/p}]\\
    &\leq \bbE_{\bg_1, \bg_2}[\tr(\bbE_s[O_s(\bh_0)]^p)]^{1/p}\\
    &\leq O\left(d \left(\frac{4p}{2^r}\right)^{p/2} + d^p\left(\frac{2^r-1}{D}\right)\right)^{1/p}\\
    &\leq O\left(d^{1/p} \sqrt\frac{4p}{2^r} + d\left(\frac{2^r-1}{D}\right)^{1/p}\right), 
\end{align*}
where we used the fact that the spectral norm is upper bounded by every Schatten $p$-norm in the first line, Jensen's inequality was applied in the second, Theorem \ref{thm:estimateforT} in the third and subadditivity in the fourth. This gives the desired result. 

That $\|O_s(\bh)\|_\infty \leq 1$ for all $\bh$ is simple to see. Finally, since $\zeta^3=e$, the identity element of $G$, multiplying $\bh$ successively by $\zeta$ cycles through $\bh,\zeta\bh,\zeta^2\bh$. Consequently, each of the three corresponding POVM elements occurs with coefficients $1/2,-1/4,-1/4$ in the sum, giving
\[
\sum_{j=0}^2 O_s(\zeta^j\bh)
=\left(\frac12-\frac14-\frac14\right)\sum_{j=0}^2 M^{\zeta^j\bh,s}_1
=0. \qedhere
\]

\end{proof}

\section{Completing the proof of Theorem \ref{thm:maintheorem}}\label{sec:heis-group}

We now complete the proof of Theorem \ref{thm:maintheorem} by supplying a concrete instantiation of the $\ShiftEq$ problem that was investigated in the previous section. We will show: 

\begin{theorem}
\label{thm:heis-group}
For all $r \geq 1$, there exists a group $G_r$ of size $|G_r|=3^{2^{2r+1}+1}$ such that in the statement of Theorem \ref{thm:quantumlowerbound}, one may take $D$ to be~$3^{2^{2r}}$.
\end{theorem}

This lets us conclude \Cref{thm:maintheorem}.
\begin{proof}[Proof of Theorem \ref{thm:maintheorem} assuming Theorem \ref{thm:heis-group}]
    For $r \geq 1$, consider the communication problem $\ShiftEq_{G_r, r}$, with $G_r$ as in the statement of Theorem \ref{thm:heis-group}. 
    
    By \Cref{thm:reductiontoMEMB,thm:upperboundsubgroupmembership}, there is a one-way entanglement-assisted protocol for $\ShiftEq_{G_r, r}$ with $\Theta(r)$ bits of communication. Letting $\tG_r$ be the group promised by the statement of Theorem \ref{thm:reductiontoMEMB}, we see that $\log |\tG_r| = \Theta(2^{3r})$ shared EPR pairs suffice to perform the protocol promised by Theorem \ref{thm:upperboundsubgroupmembership}. On the other hand, in the entanglement-unassisted setting, Theorem \ref{thm:quantumlowerbound} gives:
    \[Q^{1, \text{pub}}_{1/3}(\ShiftEq_{G_r, r}) = \Omega(\min \{ 2^r, 2^{2r} (\log 3) / 2^r\}) = \Omega(2^r).\]
    A matching upper bound of $O(2^r)$ is provided by the unassisted classical SMP protocol for bounded-order subgroup membership in \Cref{thm:subgroup_classical_other}, combined with Theorem~\ref{thm:reductiontoMEMB}, using that the subgroup size is $k = 2^{r+1}$.
    
    We compute the length of Alice and Bob's inputs when represented as bit strings:
    \[\log |G^S \times G^S| = 2 |S| \log|G| = \Theta(2^{3r}), \qquad \text{and} \]
    \[\log |G^S \times S| = |S|\log|G| + \log |S| = \Theta(2^{3r}). \]
    Let $n_r$ be the maximum of the two input lengths above, rounded up to an integer. By padding if necessary, choose surjective maps from $\{0,1\}^{n_r}$ onto Alice's and Bob's respective input sets, and define $f_{n_r}$ by applying $\ShiftEq_{G_r,r}$ to the decoded inputs. Since the maps are surjective, the communication lower bounds for $\ShiftEq_{G_r,r}$ also apply to $f_{n_r}$, while the upper bounds are unchanged. Finally, standard padding defines $f_n$ for all $n$, without changing the asymptotic bounds. This finishes the proof of Theorem \ref{thm:maintheorem}.
\end{proof}


We return to the construction of a suitable family of groups as required for \Cref{thm:heis-group}.

\begin{proof}[Proof of \Cref{thm:heis-group}]
For $m\geq 1$, the generalized Heisenberg group $\rH_m$ over a field $F$ has group elements of the form $(u, v, w)$, where $u,v \in F^m$, and $w\in F$. Its group operation is defined as
\[
(u_1, v_1, w_1)(u_2, v_2, w_2) = (u_1 + u_2, v_1 + v_2, w_1 + w_2 + u_1 \cdot v_2).  
\]
where $x \cdot y = \sum_{i = 1}^m x_i y_i $ is the standard scalar product over $F^m$.

To understand the group operation in $\rH_m(F)$, we look at its standard representation, which we denote by $H$, on the $(m+2)$-dimensional vector space over $F$:
\[ H((u,v,w)) = \begin{pmatrix} 1 & u^\transpose & w \\ 0 & \id_m & v \\ 0 & 0 & 1\end{pmatrix}. \]
Then,
\begin{align*}
H((u_1,v_1,w_1))H((u_2,v_2,w_2)) & = \begin{pmatrix} 1 & u_1^\transpose & w_1 \\ 0 & \id_m & v_1 \\ 0 & 0 & 1\end{pmatrix}\begin{pmatrix} 1 & u_2^\transpose & w_2 \\ 0 & \id_m & v_2 \\ 0 & 0 & 1\end{pmatrix} \\
& = \begin{pmatrix} 1 & u_1^\transpose + u_2^\transpose & w_1+w_2+u_1\cdot v_2 \\ 0 & \id_m & v_1+v_2 \\ 0 & 0 & 1\end{pmatrix} \\
& = H((u_1,v_1,w_1)(u_2,v_2,w_2)),
\end{align*}
which means that the group operation is just matrix multiplication of upper triangular matrices.

We will take $F=\bbF_3$, and consider the group $\rH_m(\bbF_3)$. Below, we omit the field underlying $\rH_m$ for ease of notation.

Let~$\zeta = (0,0,1)$. We may verify that~$\zeta \in Z(\rH_m)$ and has order~$3$. 

We bound the smallest dimension~$D(\rH_m, \zeta)$ of a unitary irrep~$\rho$ of~$\rH_m$ such that~$\rho(\zeta) \neq \id$.
\begin{lemma}
\label{lemma-irrep-bound}
$D(\rH_m, \zeta) = 3^m$.
\end{lemma}
\begin{proof}
Let~$\rho \in \hrH_m$ be a unitary irrep of~$\rH_m$ such that $\rho(\zeta)\neq I$, with dimension~$d_\rho$. As stated before, since~$\zeta$ is in the center of~$\rH_m$, the operator~$\rho(\zeta)$ commutes with~$\rho(g)$ for every~$g \in \rH_m$. Hence, by the Schur Lemma, we have~$\rho(\zeta) = \lambda \id$ for some complex number~$\lambda$. Since~$\zeta$ has order three, we have~$\lambda = \upomega^{s(\rho)}$ for some~$s(\rho) \in \{1,2\}$, where~$\upomega = \exp(2 \pi \complexi / 3)$.

For~$u,v \in \finitefield_3^m$, let~$X_u = \rho(u, 0 , 0)$ and~$Y_v = \rho(0, v, 0)$. Since
\[
(u, 0, 0)( 0, v, 0) = (u, v, u \cdot v) = (0, 0, u \cdot v) ( 0, v, 0)(u, 0, 0),
\]
we have~$X_u Y_v = \upomega^{s\, u\cdot v} Y_v X_u$. Since the elements~$(u,0,0)$ all commute, the operators~$(X_u : u \in \finitefield_3^m)$ all commute. The operators~$(X_u)$ are therefore simultaneously diagonalizable. Let~$\ket{\phi}$ be a common eigenvector of all the operators~$(X_u)$. We claim that the vectors~$\{ Y_v \ket{\phi} : v \in \finitefield_3^m \}$ are mutually orthogonal. 

Fix~$v \in \finitefield_3^m$. Observe that~$X_u Y_v \ket{\phi} = \upomega^{s\, u\cdot v} Y_v X_u \ket{\phi} = \lambda_u \upomega^{s\, u\cdot v} Y_v \ket{\phi}$, where~$\lambda_u$ is the eigenvalue of~$X_u$ corresponding to~$\ket{\phi}$. Since~$Y_v$ is unitary, we have~$Y_v \ket{\phi} \neq 0$. Hence~$Y_v \ket{\phi}$ is also a simultaneous eigenvector of~$X_u$ for all~$u$, and the corresponding eigenvalue of~$X_u$ is~$\lambda_u \upomega^{s\, u\cdot v}$.

Consider~$w \in \finitefield_3^m$ distinct from~$v$, and let~$j \in [m]$ be a coordinate in which~$v$ and~$w$ differ. Let~$e \in \finitefield_3^m$ be the element such that~$e_i = 0$ in all coordinates~$i \neq j$, and~$e_j = 1$. Then~$Y_v \ket{\phi}$ and~$Y_w \ket{\phi}$ are both eigenvectors of~$X_e$ with eigenvalues~$\lambda_e \upomega^{s v_j}$ and~$\lambda_e \upomega^{s w_j} $, respectively. Since~$v_j \not\equiv w_j \pmod{3}$ and $s \not\equiv 0 \pmod{3}$, we have~$\upomega^{s v_j} \neq \upomega^{s w_j}$. Since~$X_u$ is unitary, we have~$\abs{\lambda_u} = 1$ for all~$u$. So the eigenvectors~$Y_v \ket{\phi}$ and~$Y_w \ket{\phi}$ of~$X_e$ have distinct eigenvalues. Since~$X_e$ is unitary, the eigenvectors~$Y_v \ket{\phi}$ and~$Y_w \ket{\phi}$ are mutually orthogonal. Altogether, we conclude that the vectors~$\{ Y_v \ket{\phi} : v \in \finitefield_3^m \}$ span a space of dimension at least~$3^m$, i.e.,~$D(\rH_m, \zeta) \ge 3^m$. 

Conversely, we exhibit an irrep that has dimension~$3^m$. Let~$W = \complex^{\finitefield_3^m} = \Span \big\{ \ket{x} : x \in \finitefield_3^m \big\}$. Define a map~$\sigma : \rH_m \rightarrow \linear(W)$ by
\[
\sigma(u,v,z) \ket{x} = \upomega^{z + u\cdot x} \ket{x + v} \qquad \text{for all } x \in \finitefield_3^m,
\]
with~$\upomega$ as above. We may verify that~$\sigma$ is a group homomorphism and that~$\sigma(\zeta) = \upomega \id$. By construction, $\sigma$ has dimension~$3^m$. The operators~$(\sigma(u,0,0) : u \in \finitefield_3^m)$ are simultaneously diagonalizable in the standard basis~$(\ket{x} :  x \in \finitefield_3^m)$, with each joint eigenspace being one-dimensional. However, the operators~$(\sigma(0,v,0) : v \in \finitefield_3^m)$ act transitively on these eigenspaces, as for any~$x,x' \in \finitefield_3^m$, we have~$\sigma(0,x'-x,0) \ket{x} = \ket{x'}$. So no proper subset of the joint eigenspaces of~$(\sigma(u,0,0))_u$ is preserved by~$(\sigma(0,v,0))_v$ and the representation~$\sigma$ is irreducible. The dimension of~$\sigma$ is~$3^m$ by construction.
\end{proof}

We obtain a suitable family of groups~$G_r$ by an appropriate choice of the parameter~$m$. In order to obtain the largest possible bound from \Cref{thm:quantumlowerbound}, we would like~$(\log D(\rH_m, \zeta))/2^r \ge 2^r$. We take~$m = 2^{2r}$, and set~$G_r = \rH_m$ to conclude the theorem.
\end{proof}

\begin{remark}
The group $H_m$ we have picked above is optimal in terms of the value of $D(G,\zeta)$ relative to $|G|$. For the central $\zeta$, it can be shown $D(G,\zeta)^2 \leq \frac{|G|}{|Z(G)|}\leq \frac{|G|}{3}$ (since by the definition of $\zeta$, $Z(G)$ contains at least three distinct elements, $\zeta, \zeta^2, \zeta^3$). This inequality is saturated by the Heisenberg group.
\end{remark}

\printbibliography

\appendix

\end{document}